\documentclass[12pt,british]{article}
\usepackage[T1]{fontenc}
\usepackage[latin9]{inputenc}
\usepackage{amsmath}
\usepackage{amsthm}
\usepackage{amssymb}
\usepackage{natbib}
\usepackage{comment}
\usepackage{titling}

\makeatletter

\usepackage{amsthm}

\theoremstyle{definition}
\newtheorem{defn}{Definition}

\makeatother

\usepackage{babel}

\providecommand{\lemmaname}{Lemma}
\providecommand{\propositionname}{Proposition}

\usepackage[pagewise]{lineno}

\usepackage{setspace}
\usepackage{hyperref}
\usepackage{geometry}
\theoremstyle{plain}
\newtheorem{theorem}{Theorem}
\theoremstyle{plain}
\newtheorem{lem}{\protect\lemmaname}
\theoremstyle{plain}
\newtheorem{coroll}{Corollary}
\theoremstyle{plain}
\newtheorem{prop}{\protect\propositionname}
\theoremstyle{plain}

\newcommand{\ot}{\frac{1}{2}}

\selectfont
\begin{document}

\setlength{\droptitle}{-1cm}   

\title{\textbf{Troll Farms}\thanks{We thank James Best, Brendan Daley, Paul Heidhues, \'{A}ngel Hern\'{a}ndo-Veciana, Navin Kartik, Qianjun Lyu, Gerard Padr\'{o} i Miquel, Johannes Schneider,
Ludvig Sinander, Dana Sisak and seminar audiences in Bielefeld, Bremen, Dortmund,  Lancaster, L\"{u}neburg, Complutense Madrid, and Riga as well as audiences of CMID, EPCS, MAPE, Oligo workshop, SAEe,  SAET, BSE Summer Forum, and the UC3M Micro Retreat for valuable feedback. We also gratefully acknowledge the support from the following grants: CEX2021-001181-M by MICIU/AEI /10.13039/501100011033; PGC2018-098510-B-I00, PID2020-118022GB-I00, RYC2021-032163-I, RYC2022-036590-I, PID2022-141823NA-I00, and PID2023-151783NB-I00 by the Agencia Estatal de Investigaci\'{o}n (Spain); EPUC3M11 (V PRICIT) and H2019/HUM-5891 by the Comunidad de Madrid (Spain).}}
\author{Philipp Denter\thanks{Universidad Carlos III de Madrid, Department of Economics, Calle de Madrid 126, 29803 Getafe, Spain. E-Mail: \href{mailto:pdenter@eco.uc3m.es}{pdenter@eco.uc3m.es}.} \and Boris Ginzburg\thanks{Universidad Carlos III de Madrid, Department of Economics, Calle de Madrid 126, 29803 Getafe, Spain.  E-Mail: \href{mailto:bginzbur@eco.uc3m.es}{bginzbur@eco.uc3m.es}.}}

\date{\today}

\maketitle


\begin{abstract}

We study how coordinated disinformation campaigns affect elections. We develop a constrained information design model in which a sender deploys uninformative messages that mimic voters' exogenous informative signals. Voters initially opposed to the sender's preferred outcome receive favourable messages, while those in favour are targeted with unfavourable messages to dilute adverse information. The sender's ability to manipulate political outcomes increases with greater precision of voters' independent signals, but decreases with polarisation. When messaging is costly, the sender may stop targeting marginally opposing voters while moderating message extremism among supporters.



\end{abstract}



\newpage

\newpage

\section{Introduction}

\begin{flushright}
``\emph{What we're facing is a new form of propaganda that\\ wasn't really possible until the digital age.}''
\end{flushright}
\begin{flushright}
Sean Illing, Vox
\end{flushright}

\begin{flushright}
``\emph{The real opposition
is the media. And the way to deal with \\ them is to flood the zone with shit.}''\\
\end{flushright}

\begin{flushright}
Attributed to Steve Bannon
\end{flushright}

\bigskip

%

%
%
Voters receive a significant amount of information through social media platforms.
This enables them to acquire information from many sources, but also gives malicious agents new tools to manipulate information.
One growing concern are the so-called troll farms -- groups of coordinated social media accounts that disseminate biased messages purporting to come from real individuals or reliable news sources.
Such troll farms are increasingly common: one report estimates that in 38 out of 65 surveyed countries, political leaders use them to manipulate elections and other domestic political events \citep{freedomhouse2019crisis}.
For example, the use of troll farms has been detected in the 2016 Brexit referendum and the 2016 US presidential election \citep{gorodnichenko2021social}, as well as in the online debate in China \citep{king2017chinese}.


Several features of troll farms make them more effective than more traditional propaganda tools such as biased media. First, unlike messages delivered via traditional channels, messages from troll farms shared via social media are able to emulate genuine information \citep{ferrara2016rise,haustein2016tweets}. Thus, the target audience is left uncertain whether a particular message comes from a real person communicating a piece of news or an experience, or from a troll. 

Second, social media platforms collect vast amounts of personal data, enabling them to infer users' political preferences. This makes it possible to target specific voters with specific messages designed to maximise persuasive effect.\footnote{For example, already in 2012 political campaigns could deliver specific ads to, for example, ``all registered 50-to-60-year-old male Democrats in Pennsylvania's 6th district who are frequent voters and care about the environment.'' \citep{Leber:2012}. 
} The impact of this microtargeting capacity of social media influence operations
has been demonstrated in  2018, when the consulting firm Cambridge Analytica was found to have collected data on millions of Facebook users to determine their political preferences and target them with custom-made political messages \citep{NYT:Cambridge}. 
The effectiveness of such persuasive microtargeting has  been documented by \cite{SimchonEtAl:2024}.


Third, modern technology allows large numbers of trolls to be deployed at almost no cost, as multiple fake accounts can be controlled by a single user or even by automated algorithms \citep{guardian_sock_puppet_software,freedomhouse2017automated}.
As a result, troll farms are producing a substantial proportion of content in online political discussions.\footnote{During the first impeachment of US president Donald Trump, an estimated 31\% of impeachment related tweets were coming from bots \citep{rossetti2023bots}. In the early weeks of the Russia-Ukraine war, bot accounts produced around 17\% of tweets related to the conflict \citep{shen2023examining}.}
The emergence of generative language models such as ChatGPT can make it particularly easy to flood social media platforms with content generated by fake accounts \citep{goldstein2023generative}, potentially drowning out other messages.

In this paper, we analyse the impact of troll farms on voting outcomes. We develop a constrained information design framework that reflects the three aforementioned features of troll farms. In our model, a continuum of voters choose between two actions. To fix ideas, we will say that voters choose whether to vote for or against the government. There is a binary state of the world. Voters share a common prior about the state, and each voter independently receives an imperfect continuous signal about it. Moreover, voters differ in their political preferences, or \emph{types}. 

A sender wants to increase the share of voters voting for the government. She can do it by garbling the information that voters receive.
Specifically, the sender can send uninformative messages that mimic exogenous informative signals. Each voter receives a message, and does not  know whether it is an informative signal or comes from the troll farm. These messages can be targeted at specific voters. For each voter type, the sender chooses the \emph{intensity} of trolls' activity and its \emph{slant} -- that is, the probability that the voter  receives a message from the troll farm instead of a genuine informative signal, and the distribution of the trolls' messages.

We then characterise the optimal strategy of the sender.
To voters who prefer to vote against the government at the prior belief, the sender sends pro-government messages to increase the probability that such voters receive them. On the other hand, voters that prefer to vote in favour of the government at the prior belief are targeted with anti-government messages to make them less credible.

The result that the sender sends both pro-government and anti-government messages may appear counterintuitive. It is, however, in line with some of the empirical literature on troll farms, which concludes that troll farms operated by the same sender send messages of opposite political leanings \citep{linvill2020troll}. One common interpretation for this observation is that troll farms aim to create conflict rather than persuade voters to support a particular side of the political spectrum \citep{simchon2022troll,linvill2022talking}. Our result suggests that the persuasion motive is also consistent with this empirical pattern.

We then derive several results related to the ability of the sender to manipulate electoral outcomes. First, we show how precision of voters' independent signals affects electoral outcomes and the sender's payoffs. In the absence of trolls, making voters' independent signals more informative increases the share of voters voting for the government in the high state, and reduces it in the low state. However, the presence of trolls changes this picture: Theorem \ref{prop: informativeness_Blackwell} shows that higher precision of voters' independent signals   increases the share of the voters who back the government in \emph{both} states. Thus, societies in which voters have access to high-quality information online are more vulnerable to troll farms. At the same time, we show that even though troll farms garble information, their presence can help move the election towards the efficient outcome when signals are moderately informative.

Theorem \ref{prop: informativeness_Blackwell}  contrasts with the literature on information design without source uncertainty, in which more precise independent signals make it harder for the sender to manipulate  beliefs \citep{BergemannMorris:2016,dentersocial,gradwohl2022social}. The reason is that in that literature, voters can distinguish between the sender's messages and independent signals. Hence, the sender is constrained to information structures that are more informative than those induced by independent signals. Increased informativeness of independent signals tightens this constraint. In our model, on the other hand, the signal of the sender mimics voters' private signals rather than complementing them. Hence, the sender faces the opposite constraint: she is restricted to information structures that are less precise than those induced by independent signals. Making exogenous signals more precise relaxes this constraint, increasing the sender's ability to persuade voters.

Second, we analyse how increased polarisation of the electorate affects the power of the troll farm. While there is considerable discussion of potential negative effects of increased polarisation \citep{mccoy2018polarization,martherus2021party}, we show in Theorem \ref{prop: polarisation} that increased polarisation limits the ability of troll farms to manipulate electoral outcomes. We also show that deviations from Bayesian inference that cause voters to underreact to information have a similar effect. Intuitively, polarisation of voters' preferences means that voters tend to require stronger signals to be persuaded, which has the same aggregate effect as making individual signals less precise. 

Third, we show that when it is costly to send trolls' messages to voters, the sender adjusts the strategy differently depending on the voters' types. For pro-government voters, the sender continues to target all voters but adjusts the slant, ceasing to send extreme anti-government messages. For anti-government voters, the sender keeps the slant unchanged but ceases to target moderates.

\paragraph{Related literature.}

The paper adds to the literature on constrained information design. While the basic Bayesian persuasion framework \citep{KamenicaGentzkow:2011} allows the sender to commit to any Bayes-plausible information structure, a number of papers have examined settings in which the sender faces constraints on the set of information structures available to her.\footnote{These constraints can come from information being costly for the sender to acquire \citep{di2021strategic}, from receiver's ability to detect lies \citep{ederer2022bayesian}, from communication being  noisy \citep{le2019persuasion}, from  gradual communication \citep{escude2023slow}, from the sender being  restricted to monotone partitions of the state space \citep{onuchic2023conveying}, or to censorship strategies \citep{ginzburg2019optimal}.}
The closest papers in this literature are those that study persuasion of receivers who, as in our paper, observe independent private signals in addition to a signal from the sender \citep{BergemannMorris:2016,dentersocial,Matyskova:2018,heese2025persuasion,gradwohl2022social}. In contrast, in our paper the sender imitates independent signals rather than complementing them, which underlies Theorems \ref{prop: informativeness_Blackwell} and \ref{prop: polarisation}.

Our paper also contributes to the literature that  models propaganda and censorship   as information design problems. 
\cite{gehlbach2014government} and \cite{gitmez2023informational} analyse Bayesian persuasion of citizens by an autocrat. \cite{shadmehr2015state}, \cite{boleslavsky2021media} model rulers who choose the extent to which they restrict the ability of media to transmit information. In \cite{kolotilin2022censorship}, a government faces media outlets with different reporting strategies, and chooses which ones to suppress.
\cite{gitmez2023dictator} study a dictator who chooses a persuasion scheme together with a level of repression.  Our paper adds to this literature by analysing targeted political persuasion with source uncertainty.

At the same time, our results  contribute  to the growing literature on disinformation on social media. A number of papers have studied the spread of misinformation online (see \citealp{zhuravskaya2020political} for an overview). 
In particular, prior research has documented extensive use of troll farms in the 2016 Brexit referendum and the 2016 US presidential election \citep{gorodnichenko2021social}; as well as in the online debate in China \citep{king2017chinese}. Recent theoretical models of online misinformation include \cite{foerster2023theory}, \cite{acemoglu2024model}, \cite{aymanns2025fake}, and \cite{SisakDenter:2025}.
Our paper complements this  literature by analysing a sender that uses trolls to manipulate electoral outcomes.

Another literature models political persuasion as signal jamming \citep{edmond2013information,caselli2014incumbency,akoz2016information,edmond2021creating}. These models share our premise that a sender distorts an exogenous signal to manipulate receivers, but differ from our framework in two important respects. First, the sender can shift only the mean of the signal distribution, whereas our framework allows manipulation of the entire distribution, enabling a richer characterisation of optimal targeting strategies. Second, unlike in signal jamming models,  we adopt an information design approach in which  the sender commits to a strategy. This  generates the key implication that the sender does not wish to fully replace informative signals even when doing so is costless.



The paper also contributes to the  literature studying microtargeting in political persuasion. In the last few years, evidence has been mounting that automated microtargeting is very precise \citep{YouyouEtAl:2016} and highly effective \citep{ZaroualiEtAl:2022,TappinEtAl:2023}. \cite{SimchonEtAl:2024} show how generative artificial intelligence tools such as ChatGPT or Microsoft Copilot allow precise and effective microtargeting on a large scale at low cost. Theoretical papers studying microtargeting of voters include \cite{SchipperWoo:2019} and \cite{vanGils2025microtargeting}, which show the effects of microtargeting on voter welfare.
A number of papers also analyse senders targeting  voters in a pivotal-voter framework \citep{bardhi2018modes,chan2019pivotal}.

The rest of the paper is structured as follows. Section \ref{sec:model} introduces the model and discusses its main assumptions. Section \ref{sec:optimal_strategy} characterises the optimal strategy of the sender and the electoral outcomes with trolls. Section \ref{sec:main results} introduces the main results of the paper related to the effects of signal informativeness and electoral polarisation on the outcomes. Section \ref{sec:extensions} explores several extensions to the main model. Finally, Section \ref{sec:discussion} concludes. All proofs are in the Appendix.

\section{Model}\label{sec:model}

A continuum of voters of mass one need to choose whether to re-elect the government. There is an unknown state of the world $\theta\in\left\{ 0,1\right\} $, which indicates, for example, whether the government is competent.
The preferences of each voter are characterised by a type $x\in\left[0,1\right]$. Types are distributed across voters according to a CDF $H\left(x\right)$ with the associated density $h\left(x\right)$ and with full support.
If a voter with type  $x$ votes for the government, she receives a payoff $1-x$ if the state turns out to be 1, that is, if the government is competent, and a payoff $-x$ if the state turns out to be 0. The payoff of a voter who votes against the government is normalised to zero.
Thus, a voter with a higher type is more opposed to the government. We assume that a voter who is indifferent votes for the government. The government is re-elected if the share of voters who vote for it is at least $\frac{1}{2}$.

Voters share a common prior belief about the state of the world being 1. We normalise that belief to $\frac{1}{2}$. Note that this is without loss of generality: as we show below, the voter votes for the government if and only if her posterior belief is at least $x$. Hence, changing the prior belief is equivalent to changing the distribution $H$. 

At the beginning of the game, each voter $i$ receives a private signal $s_{i} \in \mathbb{R}$. In each state $\theta\in\left\{ 0,1\right\} $, the signal of each voter is drawn from a CDF $F_{\theta}$ with density $f_{\theta}$ and full support on $\mathbb{R}$. We assume, without loss of generality, that $f_{0}\left(0\right)=f_{1}\left(0\right)$. Let $m\left(s\right):=\frac{f_{1}\left(s\right)}{f_{0}\left(s\right)}$ denote the likelihood ratio. We will assume that $m\left(s\right)$ is strictly increasing in $s$. This imposes a standard monotone likelihood ratio assumption, implying that a higher realisation of the signal makes a voter believe that the state equals one with a higher probability. It also implies that $F_{0}\left(s\right)>F_{1}\left(s\right)$ for all $s$.
We will refer to the pair $\left(F_{0},F_{1}\right)$ as the \emph{information structure}.

A player, whom we will call the sender, is trying to help the government. She can do it by setting up a troll farm, that is, by sending messages that imitate and replace the informative signals but are not correlated with the true state. Hence, the sender cannot supply additional information to voters, and is constrained to garblings of their exogenous signals. For each type of voter, the sender can select the \textit{intensity} with which she is targeted by trolls, as well as the  \textit{slant} of the trolls' messages. Formally, for each type of voter $x$, intensity equals  the probability $\alpha_{x}$ with which the voter observes a signal from trolls instead of an informative signal. For example, $\alpha_{x}=0$ means that no trolls are targeting voters with type $x$, and thus all signals observed by these voters are coming from informative sources. On the other hand, $\alpha_{x}\rightarrow1$ means that a signal is almost surely coming from a troll. Slant is represented by a distribution $\tilde{F}_{x}$ of the trolls' signals, with the associated density $\tilde{f}_{x}$.
Informally, trolls have a more pro-government slant if $\tilde{F}_{x}$ is such that messages tend to have high realisations.


As will be shown later, there may be several strategies that result in the same share of votes for the government. As an equilibrium selection rule, we will assume that, all other things being equal, the sender prefers the smallest intensity.
Formally, suppose that choosing
$\left(\alpha_{x},\tilde{F}_{x}\right)_{x \in \mathbb{R}}$ and $\left(\alpha_{x}^{\prime},\tilde{F}_{x}^{\prime}\right)_{x \in \mathbb{R}}$ produce  the same vote share for the government in both states, but $\alpha_{x}\leq\alpha_{x}^{\prime}$ for \emph{all} $x$, with strict inequality for some $x$. Then the sender prefers $\left(\alpha_{x},\tilde{F}_{x}\right)_{x \in \mathbb{R}}$ to $\left(\alpha_{x}^{\prime},\tilde{F}_{x}^{\prime}\right)_{x \in \mathbb{R}}$.

The timing of the game is as follows. First, for each $x$, the sender selects intensity $\alpha_{x}$ and slant $\tilde{F}_{x}$. Then, nature draws the state $\theta$ and the signals. Each voter then receives either a signal from the troll farm or an informative signal, without being able to distinguish between the two. 
Voters then form their posterior beliefs and vote. Afterwards, payoffs are realised. The payoff of each voter is as described above, while the payoff of the sender is $u\left(V\right)$, where $V$ is the share of voters that vote for the government, and $u$ is a strictly increasing function.
\paragraph{Discussion of assumptions.}

The model introduced above incorporates the three features of troll farms discussed in the introduction. First,  a voter who observes a signal is unable to distinguish between a message from a troll and a genuine message. This reflects the fact that trolls mimic informative signals. 
Second, the sender is able to microtarget voters by selecting  a different intensity and slant for each type of voter. This represents the ability of troll farms to use information about voters' political preferences that social media platforms are able to extract. 
Third, because trolls can produce a large number of signals cheaply, increasing intensity is costless for the sender in the model. In Section \ref{subsec:limited reach} we  discuss how the results of the model change when sending trolls' messages is costly.

When several strategies give the sender the same payoff, the sender selects the one with the lowest intensity. Section \ref{subsec:limited reach} micro-founds this assumption, showing that this equilibrium selection rule emerges as a limit case when the cost of trolls is positive but goes to zero. At the same time, as subsequent results show, this assumption is important for describing the optimal strategy of the sender, but not for the comparative statics captured in Theorems \ref{prop: informativeness_Blackwell} and \ref{prop: polarisation}.

Our model describes a setting in which voters are aware of the existence of the troll farm, and of its chosen intensity or slant. In Section \ref{sec:naive} we consider a case in which a subset of voters are naive and do not realise that a troll farm is operating.

The model also assumes that the sender is able to replace all information that a voter can acquire. In reality, some information is acquired through channels other than social media, and hence is beyond the sender's control. In Section \ref{sec:Non-Replaceable}, we show that the main results of the paper remain unchanged if some information comes from sources that cannot be emulated by trolls.

By replacing informative signals with uninformative ones, the sender is able to garble the information that voters receive. Hence, she is able to induce distributions of posterior beliefs that are less accurate than those induced by these independent signals. However,  she cannot send informative signals of her own, and hence cannot induce distributions of posterior beliefs that are more precise. If she could both mimic exogenous signals and commit to an arbitrary informative experiment,  she could induce any  Bayes-plausible distribution of posteriors by setting $\alpha_{x}=1$ and fully replacing the exogenous signals. Thus, it would be optimal to replace all independent signals, and the model would become a standard model of Bayesian persuasion.

Another feature of our model is that the sender cannot observe the state of the world. Suppose instead that she is informed about the state. 
In this case she can condition $\left(\alpha_{x},\tilde{F}_{x}\right)$  for each voter type on it. Then $\left(\alpha_{x},\tilde{F}_{x}\right)$ serves as a signal about the state, and the setting becomes that of a cheap talk game. As usual, there exists a babbling equilibrium, in which the sender selects the same  $\left(\alpha_{x},\tilde{F}_{x}\right)$ in each state. This equilibrium is identical to ours.
On the other hand, a separating equilibrium cannot exist because if the sender's choice reveals information about the state, she will deviate in state 0. Because only a babbling equilibrium exists, allowing the sender to observe the state does not change the outcome of the game.

We also assume that voters aim to match their actions to the state, but do not receive a payoff from the outcome of the vote. The fact that voters care about their actions may represent, for example, expressive motivations for voting, for which there is empirical evidence \citep{feddersen2009moral,pons2018expressive,ginzburg2022counting}. Since the set of voters is a continuum, a given voter is pivotal with probability zero, so making the voters' utility also dependent on the outcome of the vote does not change their equilibrium actions, as long as each voter puts a positive weight on her action.


\section{Optimal Intensity and Slant}\label{sec:optimal_strategy}
\subsection{A Benchmark without Trolls}\label{sec:notrolls}

Consider first what happens when trolls are not available. Take a voter with type $x$. Let

\[
\pi(s)=\frac{f_{1}\left(s\right)}{f_{1}\left(s\right)+f_{0}\left(s\right)}=\frac{m\left(s\right)}{m\left(s\right)+1}
\]
be the probability that a voter assigns to the government being competent when she observes signal $s$. Her expected payoff is $\pi\left(s\right)-x$ if  she votes for the government, and zero otherwise. Hence, she votes for the government if and only if $\pi\left(s\right)\geq x$, or, equivalently, if and only if $s\geq s^{*}\left(x\right)$, where $s^{*}\left(x\right):=m^{-1}\left(\frac{x}{1-x}\right)$. Therefore, in each state $\theta \in \left\{0,1\right\}$, a voter with type $x$ votes for the government with the ex ante probability $1-F_{\theta}\left[s^{*}\left(x\right)\right]$. This probability depends on the information structure, which also affects the cutoff $s^{*}\left(x\right)$ through the likelihood ratio.

We say that signals are more informative if each voter is more likely to make the correct decision in each state. 
Formally, we define greater informativeness as follows:


\begin{defn}\label{def: informativeness}
Information structure $\left(\hat{F}_{0},\hat{F}_{1}\right)$ is more informative than information structure $\left(F_{0},F_{1}\right)$ if and only if for all $x$ we have
$\hat{F}_{0}\left[\hat{s}^{*}\left(x\right)\right]\geq F_{0}\left[s^{*}\left(x\right)\right]$ and
$\hat{F}_{1}\left[\hat{s}^{*}\left(x\right)\right]\leq F_{1}\left[s^{*}\left(x\right)\right]$.
\end{defn}


Note that if $\left({F}_{0},{F}_{1}\right)$ is less informative than $\left(\hat{F}_{0},\hat{F}_{1}\right)$ according to this definition, it is also a garbling of $\left(\hat{F}_{0},\hat{F}_{1}\right)$, and less informative in the Blackwell sense.

Given the distribution of signals in each state, the share of voters voting for the government in state $\theta\in\left\{ 0,1\right\} $ when there are no trolls equals
\[
V^{NT}_{\theta}=\int_{0}^{1}\left(1-F_{\theta}\left[s^{*}\left(x\right)\right]\right)dH\left(x\right).
\]
Hence, increased informativeness decreases the sender's payoff in state $0$ and increases it in state $1$. 



\subsection{Equilibrium with Trolls}\label{sec:equilibrium}

Now consider the case when trolls are available. Take a voter with type $x$. If the sender  chooses $\alpha_{x}$ and $\tilde{F}_{x}$, this voter forms the following posterior belief after observing signal $s$:
\[
\pi_{x}\left(s\right)=\frac{\left(1-\alpha_{x}\right)f_{1}\left(s\right)+\alpha_{x}\tilde{f}_{x}\left(s\right)}{\left(1-\alpha_{x}\right)f_{1}\left(s\right)+\alpha_{x}\tilde{f}_{x}\left(s\right)+\left(1-\alpha_{x}\right)f_{0}\left(s\right)+\alpha_{x}\tilde{f}_{x}\left(s\right)}.
\]
The voter  votes for the government whenever $\pi_{x}\left(s\right)\geq x$. For all $s$, a larger $\alpha_{x}$ or a larger $\tilde{f}_{x}\left(s\right)$ 
moves the belief towards the prior. 

Recall that without trolls,  a voter will vote for the government if and only if she receives signal $s\geq s^{*}\left(x\right)$. Consider first a voter with type $x\leq\frac{1}{2}$. When that voter's belief equals the prior, she is willing to vote for the government. By instructing trolls to send her signals below $s^{*}\left(x\right)$, the sender can weaken these signal realisations. If intensity is high enough,
the belief after each signal $s< s^{*}\left(x\right)$ will be weakly above $x$. At the same time, if trolls do not send signals $s\geq s^{*}\left(x\right)$, the belief after such signals will also be above $x$.
Hence, the sender is able to ensure that the voter votes for the government after \textit{any} signal realisation. To do so with the smallest intensity, the sender sets
$\tilde{f}_x\left(s\right)=0$ for all $s\geq s^{*}\left(x\right)$, while for all $s< s^{*}\left(x\right)$ she sets $\tilde{f}_x\left(s\right)$ in such a way that the belief after any such signal is exactly $x$. 




Now consider a voter with type $x>\frac{1}{2}$. This voter is ex ante opposed to the government. Since $s<s^{*}\left(x\right)$ cannot persuade this voter to vote for the government even without trolls, it cannot make her willing to vote for the government with trolls, either. It is then optimal for the sender to minimise the probability that these voters receive such a signal. Thus, the sender sets $\tilde{f}_x\left(s\right)=0$ for all $s<s^{*}\left(x\right)$. 
On the other hand, having trolls send messages $s\geq s^{*}\left(x\right)$ increases the probability that the voter receives signals that, in the absence of trolls, induce her to vote for the government. However, doing this reduces the posterior beliefs after such signals. Hence, it is optimal for the sender to set $\alpha_{x}$ and $\tilde{f}_x\left(s\right)$ such that the posterior belief after observing signal $s\geq s^{*}\left(x\right)$ equals exactly $x$. 


The next proposition formalises this intuition:

\begin{prop}
\label{lem: strategy}
The support of trolls' messages is $\left(-\infty,s^*\left(x\right)\right]$ for all $x\leq\frac{1}{2}$, and $\left[s^*\left(x\right),\infty\right)$ for all $x>\frac{1}{2}$. The belief of every voter after receiving a message from the support of troll's messages is  $\pi_x\left(s\right)=x$. Furthermore, $V_{1}>V_{0}\geq H \left(\frac{1}{2}\right)$, and $V_{\theta}>V^{NT}_{\theta}$ for each $\theta$, where $V_{\theta}$ is the equilibrium share of voters voting for the government in state $\theta$. 





\end{prop}

The sender thus chooses different communication strategies for voters who initially support the government and for those who initially oppose it. Pro-government voters are targeted with anti-government messages to make unfavourable information unreliable. Anti-government voters are targeted with pro-government messages to increase the probability of them receiving favourable information.
Pro-government voters always vote for the government, while anti-government voters vote for the government if and only if they receive a signal $s\geq s^{*}\left(x\right)$. Because signals are informative, the share of voters voting for the government is greater in state 1. Furthermore, being able to use trolls increases the sender's payoff in each state.

The result that the sender sends anti-government messages to pro-government voters may appear counterintuitive. In Section \ref{sec:naive}, we show that this may change when sufficiently many voters are naive, that is, unaware of the presence of trolls.


When it comes to the intensity, the sender places the most effort on persuading neutral voters, that is, voters whose types are close to $\ot$:

\begin{coroll}
\label{cor:dadx}
Optimal intensity $\alpha_x^*$  increases in $x$ if $x<\ot$ and decreases in $x$ if $x>\ot$. 
\end{coroll}

\section{Informativeness, Polarisation, and Voting Outcomes}\label{sec:main results}

In this section we analyse the effects of different features of the voting process on election outcomes.
We start by looking at the effects of informativeness. Recall from Section \ref{sec:notrolls} that without interference from the sender, under a more informative information structure, the government receives more votes in state 1, and fewer votes in state 0. With trolls, however, the dynamic is different, as the next result shows:

\begin{theorem}
\label{prop: informativeness_Blackwell}If information structure $\left(\hat{F}_{0},\hat{F}_{1}\right)$ is more informative than information structure $\left(F_{0},F_{1}\right)$, then in each state the payoff of the sender is higher under $\left(\hat{F}_{0},\hat{F}_{1}\right)$ than under  $\left(F_{0},F_{1}\right)$.
\end{theorem}

Hence, a more informative information structure helps the sender manipulate the voters' decisions, increasing the probability that a voter votes incorrectly in state $0$.
Intuitively, the sender cannot send informative signals of her own but can only replace voters' informative signals with uninformative ones. Hence, the beliefs she can induce in voters are constrained to be closer to the prior than those induced by informative signals. When the latter become more informative, the sender's constraint is relaxed. Hence, the set of beliefs she can induce expands, increasing her payoff.


Next, consider  effect of the distribution $H$ of voters' types. Recall that a voter with type $\frac{1}{2}$ is indifferent between supporting and opposing the government at the prior belief. More generally, voters whose types are closer to $\frac{1}{2}$ need a weaker signal to be convinced to shift their vote. Thus, for any two voters  whose types are on the same side of  $\frac{1}{2}$, the voter whose type is closer to  $\frac{1}{2}$ is more moderate. We can then say that the electorate becomes more polarised when the fraction of moderates decreases and the fraction of extreme voters increases. More formally, we use the following definition of polarisation:

\begin{defn}
\label{def: polarisation}
A distribution $\hat{H}$ dominates distribution $H$ in polarisation if and only if
$\hat{H}\left(x\right)\geq H\left(x\right)$ for all $x\leq\frac{1}{2}$, and
$\hat{H}\left(x\right)\leq H\left(x\right)$ for all $x\geq\frac{1}{2}$.
\end{defn}

This defines a partial order on the set of distributions. Intuitively, Definition \ref{def: polarisation} says that polarisation is higher if, for any type $x$, there are more voters who are on the same side of $\frac{1}{2}$ as $x$, but more extreme. This definition of polarisation is a special case of the definition in \cite{ginzburg2024comparing}. The next result shows that increased polarisation makes the sender worse off:


\begin{theorem}
\label{prop: polarisation} If distribution $\hat{H}$ dominates distribution $H$ in polarisation, then the payoff of the sender in each state is lower under $\hat{H}$ than under $H$.
\end{theorem}

Intuitively, increased polarisation means that voters tend to have more extreme preferences. This implies that they need stronger signals to be convinced to change their vote. Hence, increased polarisation has a similar effect to making signals less informative.

Finally, we can analyse how interaction of informativeness and polarisation affects the ability of the election to aggregate information, that is, to ensure that the government receives at least 50 percent if and only if the state is 1. If $H\left(\frac{1}{2}\right)\geq\frac{1}{2}$, Proposition \ref{lem: strategy} implies that the government wins the election in both states regardless of informativeness. Consider now the more interesting case when $H\left(\frac{1}{2}\right)<\frac{1}{2}$. Making the information structure more informative implies increasing $\frac{F_{0}\left[s^{*}\left(x\right)\right]}{F_{1}\left[s^{*}\left(x\right)\right]}$ for all $x$. One can define sequences of information structures along which $\frac{F_{0}\left[s^{*}\left(x\right)\right]}{F_{1}\left[s^{*}\left(x\right)\right]}$ increases from one to infinity. We then have the following result:

\begin{prop}
\label{prop: informativeness}Suppose that $H\left(\frac{1}{2}\right)<\frac{1}{2}$. Take a sequence of information structures $\left(F_{0}^{r},F_{1}^{r}\right)$ indexed by $r\in\left(-\infty,\infty\right)$, such that informativeness increases with $r$, $\lim_{r\rightarrow-\infty}\frac{F_{0}^{r}\left(s\right)}{F_{1}^{r}\left(s\right)}=1$ for all $s$, and $\lim_{r\rightarrow\infty}\frac{F_{0}^{r}\left(s\right)}{F_{1}^{r}\left(s\right)}=\infty$ for all $s$. Then there exist $r^{\prime},r^{\prime\prime}$ with $r^{\prime}<r^{\prime\prime}$, such that:
\begin{itemize}
\item for $r<r^{\prime}$, the government loses the election in both states;
\item for $r\in\left(r^{\prime},r^{\prime\prime}\right)$, the government wins the election in state $1$ only;
\item for $r>r^{\prime\prime}$, the government wins the election in both states whenever the polarisation of the distribution of types is sufficiently low; and wins the election in state 1 only otherwise.
\end{itemize}
\end{prop}

Intuitively, when signals are not very informative, they cannot induce sufficiently many voters to move their belief sufficiently far from the prior, so the government cannot get enough votes even in state 1. By Theorem   \ref{prop: informativeness_Blackwell}, increased informativeness raises the government's vote share in each state, and by Proposition \ref{lem: strategy}, that share is higher in state 1 than in state 0. Hence, for moderately high informativeness, the government  wins the election in state 1 only, so the election aggregates information. If signals are very informative and polarisation is not too high, the government wins the election in both states. On the other hand, if signals are very informative but polarisation is high, the election still aggregates information.



Proposition \ref{prop: informativeness} means that when polarisation is relatively low, an increase in informativeness can harm information aggregation. At the same time, when signals are very informative, information aggregation can be restored by increased polarisation.


Note also that the election can aggregate information in the presence of a troll farm. Together with the fact that the troll farm increase the government's vote share in each state, this means that at intermediate levels of informativeness of independent signals, the presence of the troll farm can have a positive effect on information aggregation:

\begin{coroll}
\label{cor:trolls-efficiency}
Suppose that $H\left(\frac{1}{2}\right)<\frac{1}{2}$. There exists an information structure under which the government always loses the election in the absence of trolls, and wins the election in state 1 only in presence of trolls.
\end{coroll}


\section{Extensions}
\label{sec:extensions}
\subsection{Costly Trolls}
\label{subsec:limited reach}
The preceding analysis has assumed that targeting voters with messages from trolls is costless. However, in some settings this assumption may not be appropriate. In this section, we analyse what happens if increasing intensity \(\alpha_{x}\) carries a cost.

For example, there may be a cost \(c>0\) of sending a message through another troll. If the mass of informative signals is normalised to one, and the mass of trolls' signals targeting a voter of type \(x\) is \(\mu_{x}\), then \(\alpha_x=\frac{\mu_x}{1+\mu_x}\), so the cost of replacing the voter's signals equals \(\frac{\alpha_x}{1-\alpha_x}c\), which is increasing and convex in intensity  $\alpha_x$.

In other settings, there may be an exogenous upper limit \(\hat{\alpha}\) on \(\alpha_{x}\). This could arise because platforms can detect and delete troll accounts if too many signals originate from trolls. Alternatively, a fraction \(1-\hat{\alpha}\) of each voter type may be beyond the reach of trolls because they do not use social media, or because they are sufficiently sophisticated to distinguish between a genuine signal and one originating from a troll farm. In this case, the cost is zero for \(\alpha_{x}\leq \hat{\alpha}\), and it becomes very large for \(\alpha_{x}>\hat{\alpha}\).

More generally, suppose that choosing intensity  \(\alpha_x\)  imposes a cost of $c\left(\alpha_x\right)$ on the sender, where \(c\) is an increasing, strictly convex, and continuously differentiable function. This is consistent with the first example of a cost that is linear in the mass of trolls, while the second example can be approximated as a limit of an increasing and convex function.
To keep the analysis tractable, we assume that \(u\!\left(V\right)=V\), that is, the sender aims to maximise the expected share of voters voting for the government. Then the equilibrium looks as follows:


\begin{prop}\label{prop:costly trolls}
For voter types \(x\in \left(0,\ot\right]\), the optimal intensity is strictly positive. The optimal slant has support \(\left[\hat{s}\!\left(x\right),s^{*}\!\left(x\right)\right]\), where \(\hat{s}\!\left(x\right)\leq s^{*}\!\left(x\right)\) is uniquely defined.
For voter types \(x>\ot\), the optimal intensity is strictly positive if \(x\in \left[\hat{x},1\right)\), and equals zero otherwise, where \(\hat{x}\geq\ot\) is uniquely defined. An optimal slant has support \(\left[s^{*}\!\left(x\right),\infty\right)\).
Furthermore, there exists \(\overline{c}> 0\) such that if \(c'(\alpha_x)\leq \overline{c}\) for all \(\alpha_x\leq \alpha_x^*\), then the equilibrium coincides with the one described in Proposition \ref{lem: strategy}.

\end{prop}

In words, making trolls costly reduces the amount of trolls' signals that are sent. However, the margin of adjustment differs between pro-government voters
and anti-government voters. For pro-government voters, the sender adjusts the slant by only sending moderately anti-government messages, but keeps the intensity positive. For anti-government voters, the sender adjusts the intensity by only targeting voters who are sufficiently anti-government, but may keep sending all pro-government messages.\footnote{While sending all pro-government messages is always optimal,  when $\alpha_x^*<\alpha_x^*$ the equilibrium support may also be a subset of \(\left[s^{*}\!\left(x\right),\infty\right)\).}


The last statement of the proposition implies that if the marginal cost of replacing signals is positive but sufficiently low, the equilibrium coincides with that of the baseline model. Thus, our earlier assumption that, among several strategies yielding the same payoff, the sender selects the one that minimises \(\alpha_{x}\) emerges as a unique equilibrium outcome when trolls are costly but the cost is not too high.
In this case, the results of Section \ref{sec:main results} on the effects of informativeness and polarisation remain unchanged. On the other hand, if the marginal cost of increasing intensity is very high, not sending any trolls is optimal. In this case, the equilibrium is equivalent to the benchmark case with no trolls, and greater informativeness leads to a higher probability of a socially optimal decision in each state, as in Section
\ref{sec:notrolls}.

\subsection{Non-Manipulable Information}
\label{sec:Non-Replaceable}
An important feature of our model is that all information that a voter receives is contained in the signal $s$, which trolls can mimic. Hence, by selecting intensity $\alpha_{x}$, the sender is able to replace \textit{any amount} of voter's information. In many situations, however, part of state-relevant information available to a voter cannot be replaced by a troll farm. For example, information that is received through personal experiences or offline interactions with other people, cannot be manipulated by social media bots.

Suppose that in addition to the signal  $s$ that the sender can mimic, each voter receives another signal  $\rho\in\mathbb{R}$, which cannot be mimicked or replaced. That signal represents all information that a voter receives through sources other than social media. In each state $\theta\in\left\{ 0,1\right\} $, signal $\rho$ is drawn from cumulative distribution function $G_{\theta}$ with density $g_{\theta}$.

Recall that in the baseline setting without the non-replaceable signal, a pair of distributions $\left(F_{0},F_{1}\right)$ from which signal $s$ is drawn constituted the information structure. Denote the pair $\left(F_{0},F_{1}\right)$ by $F$. With the additional signal $\rho$, the information structure that a voter faces is now $\left(F,G\right)$, where $G:=\left(G_{0},G_{1}\right)$.

Let $F^{\prime}$ be a garbling of $F$, that is, an information structure that can be obtained from $F$ by randomly drawing a signal $s^{\prime}\in\mathbb{R}$ after each observation of $s$ according to some distribution $\gamma\left(s^{\prime}\mid s\right)$. It follows from 
\cite{blackwell1951comparison} that if  $F$ is more informative than $F^{\prime}$ (in the sense of Definition \ref{def: informativeness}), then $F^{\prime}$ is a garbling of $F$. It is easy to see that replacing a signal by a message from trolls with probability $\alpha_{x}$ is a garbling of $F$. Less obvious is the reverse, namely, the fact that the payoff that the sender receives from any garbling of $F$ can be replicated using trolls. Formally:

\begin{lem}
\label{lem:garblings}Consider an information structure $\left(F,G\right)$. Let $W\left(F^{\prime},G\right)$ be the sender's expected payoff under information structure $\left(F^{\prime},G\right)$, where \textup{$F^{\prime}$} is a garbling of $F$. Then for each voter type $x$ there exists $\left(\alpha_{x},\tilde{F}_{x}\right)$ such that the sender can obtain $W\left(F^{\prime},G\right)$ by using a troll farm with strategy $\left(\alpha_{x},\tilde{F}_{x}\right)$ on the original information structure $\left(F,G\right)$.
\end{lem}

Because any strategy involving trolls implies a garbling of $F$, and any garbling of $F$ can be induced by trolls, the sender's problem is equivalent to selecting an optimal garbling of $F$. Hence, the set of belief distributions that the sender can induce is the set of belief distributions that can be induced by signal pairs ($\hat{s},\rho$), for all $\hat{s}$ that are less informative than $s$. Then if $s$ becomes more informative, the set of such belief distributions expands, which benefits the sender. Formally, we have the following result:

\begin{prop}
\label{prop:garblings_2_signals}If signal $s$ becomes more informative, the equilibrium expected payoff of the sender increases.
\end{prop}

This shows the result of Theorem \ref{prop: informativeness_Blackwell} extends to a setting in which only part of the voters' information can be replaced.

\subsection{Over- and Underreacting to Information}
\label{section:non-bayesian}

Substantial theoretical and empirical research has pointed to the fact that individuals may systematically deviate from Bayesian updating. In particular, voters can overweight or underweight signals \citep{tversky1992advances,ba2023over,augenblick2025overinference,ortoleva2024alternatives}. 
To account for such phenomena, suppose that upon observing signal realisation $s$, a voter forms her posterior belief as if she observed signal realisation $\beta\left(s\right)$, where $\beta\left(\cdot\right)$ is a strictly increasing function that distorts the perception of the signal according to a particular type of belief updating. We assume that $\beta\left(0\right)=0$. Thus, not receiving an informative signal does not lead the voter to change her belief. Underweighting signals corresponds to the case when $\beta\left(s\right)>s$ for $s<0$, and $\beta\left(s\right)<s$ for $s>0$. In this case, the posterior belief for a given $s$ is more conservative -- that is, closer  to the prior -- than the Bayesian posterior belief. On the other hand, with overweighting, $\beta\left(s\right)<s$ for $s<0$, and $\beta\left(s\right)>s$ for $s>0$, implying less conservative beliefs.

More generally, we can introduce a partial order on  belief distortion functions $\beta$ in terms of how conservative the resulting posterior beliefs are. We define it as follows:

\begin{defn}
\label{def: conservatism}A belief distortion function $\hat{\beta}$ is more conservative than belief distortion function $\beta$ if and only if
$\hat{\beta}\left(s\right)> \beta\left(s\right)$ for all $s<0$, and
$\hat{\beta}\left(s\right)< \beta\left(s\right)$ for all $s>0$.
\end{defn}

We can then show that deviations from Bayesian updating can limit the manipulative power of troll farms. Specifically, the next result shows that a more conservative belief distortion function reduces the payoff of the sender in each state:

\begin{prop}
\label{prop: nonbayesian}If belief distortion function $\hat{\beta}$ is more conservative than belief distortion function $\beta$, then the payoff of the sender in each state is lower under $\hat{\beta}$ than under $\beta$.
\end{prop}

The intuition is similar to that of Theorem \ref{prop: polarisation}: a more conservative distortion function moves the posterior belief of a given voter towards the prior. At the level of an individual voter, this has the same effect as an increase in polarisation at the aggregate level.

\subsection{Naive Voters}\label{sec:naive}

%

So far we assumed that voters understand how to rationally interpret information in the presence of the troll farm. In reality, some voters may be naive, and update their beliefs as if the troll farm did not exist. Experimental literature has shown evidence of strategic naivety in communication (\citealp{cai2006overcommunication,jin2021no}), and in voting interactions (\citealp{patty2007letting,ginzburg2019collective}).

Suppose that with probability $\phi\in[0,1]$ each voter is naive, and updates the belief thinking that all signals are informative, that is, assuming that $\alpha_x=0$. To keep the analysis tractable, let \(u\!\left(V\right)=V\), and hence  the sender aims to maximise the expected share of voters voting for the government.


If the sender can target  naive voters, then any such voter can be persuaded to vote for the government by setting $\alpha_{x}\rightarrow 1$ and choosing $\tilde{F}_{x}$ in such a way that $\tilde{F}_{x}\left[s^{*}\left(x\right)\right]=0$. Hence, at the equilibrium   every naive voter will vote for the government, while for the remaining voters, the sender chooses the strategy described in Section \ref{sec:equilibrium}. 
Thus, the comparative statics of Section \ref{sec:main results} remain unchanged.



Now suppose  that the sender cannot target   naive voters. 
Then to anti-government voters, as before, trolls only send  $s\geq s^*(x)$, as signals  $s< s^*(x)$ cause naive voters to vote against the government. Choosing intensity then involves a trade-off: increasing $\alpha_x$ means gaining votes from naive voters, but losing some votes of the non-naive voters.


For pro-government voters, sending anti-government messages ceases to be optimal, because naive voters will respond to them by voting against the government. Instead, the sender can persuade all pro-government voters to vote for the government by setting  $\alpha_x\rightarrow1$ and sending pro-government messages. Hence, for pro-government voters, the optimal slant changes but the voting outcome remains unchanged.

The next result shows that the effects of the political environment on the sender's payoff, summarised in Theorems  \ref{prop: informativeness_Blackwell} and \ref{prop: polarisation} remain unchanged in the presence of naive voters:


\begin{prop}
\label{pro:naive}
Suppose the sender cannot target naive voters. If signals become more informative, or if the electorate becomes less polarised, as defined in Definitions \ref{def: informativeness} and \ref{def: polarisation}, 
the government's vote share in each state weakly increases.
\end{prop}

The discussion above implies that even a small fraction of naive voters changes the sender's optimal slant. This may suggest that the result of Proposition \ref{lem: strategy} that pro-government voters receive anti-government messages is knife-edge. However, this is not the case: when trolls are costly and the cost is sufficiently high, setting $\alpha_x\rightarrow1$ ceases to be optimal. Hence, when both naivety and costs are present, choosing between sending anti-government messages and sending pro-government messages involves a trade-off between lower costs and higher vote share. In this case, if the share of naive voters is sufficiently low, sending anti-government messages to pro-government voters remains optimal.

\section{Conclusions}\label{sec:discussion}
%
This paper has analysed the strategy of a troll farm that aims to persuade a heterogeneous electorate to vote for the government by targeting each voter with a signal that emulates informative signals. It has shown that the optimal strategy of the troll farm differs considerably depending on whether a  voter is ex ante against or in favour of the government. For the former types of voters, the troll farm sends pro-government slanted messages in order to increase their quantity while keeping them persuasive. For the latter voters, the troll farm floods the information environment with messages with an anti-government slant in order to undermine their persuasiveness.

From the analysis it follows that features of the political environment that are normally considered to be positive can in fact reduce the efficiency of electoral outcomes when troll farms are present. In particular, more informative independent signals 
can amplify the power of troll farms. On the other hand, higher polarisation of the electorate, often considered to be a problem, can limit the manipulative power of troll farms. 
At the same time, when the median voter is biased against the government and signals are moderately informative, the presence of a troll farm can move the election towards the optimal outcome.

The framework developed in this paper can help answer normative questions about regulation of online platforms. In particular, the potential negative effect of signal informativeness, and the fact that troll farms can enhance efficiency of electoral outcomes, imply that optimal regulation may be less straightforward than might appear. We leave these questions for further research.




    


\newpage

\appendix
\section{Mathematical Appendix}

\subsection{Proof of Proposition \ref{lem: strategy}}

Take a voter with type $x\in\left[0,1\right]$. Suppose that the sender has chosen $\alpha_{x}$ and $\tilde{f}_{x}$. Then the voter's belief after observing signal $s$ equals
\begin{equation}
\label{eq:belief}
\pi_x\left(s\right)
=\frac{1}{1+\frac{\left(1-\alpha_{x}\right)f_{0}\left(s\right)+\alpha_{x}\tilde{f}_{x}\left(s\right)}{\left(1-\alpha_{x}\right)f_{1}\left(s\right)+\alpha_{x}\tilde{f}_{x}\left(s\right)}}.
\end{equation}
For the subsequent analysis, note that $\pi_x\left(s\right)$ is decreasing in $\alpha_{x}$ and in $\tilde{f}_{x}\left(s\right)$ if and only if $f_1\left(s\right)>f_{0}\left(s\right)$, that is, if and only if $s>0$. We will prove the result separately for pro-government voters (those with  $x\leq \frac{1}{2}$), and for anti-government voters (those with  $x> \frac{1}{2}$).

\paragraph{Pro-government voters.} 
\label{page:pro-gov}
Take a voter with type $x\leq\frac{1}{2}$. For this voter we have $s^{*}\left(x\right)<0$.
For any $s\in\left[s^{*}\left(x\right),0\right]$, without trolls we have $\pi_x\left(s\right)\geq x$. Since $\pi_x\left(s\right)$ is increasing
in $\alpha_{x}$ and in $\tilde{f}_{x}\left(s\right)$ for all $s<0$, we have $\pi_x\left(s\right)\geq x$ also with trolls. At the same time, for all $s>0$ we have $f_{1}\left(s\right)>f_{0}\left(s\right)$, and so $\pi_x\left(s\right)>\frac{1}{2}\geq x$ for any $\left(\alpha_{x}, \tilde{f}_{x}\left(s\right)\right)$. We can conclude that for all $s\geq s^{*}\left(x\right)$ we have $\pi_x\left(s\right)\geq x$,  so the voter votes for the government regardless of $\alpha_{x}$ and $\tilde{f}_{x}$.

At the same time, for such voters, the sender can ensure that $\pi_x\left(s\right)=\frac{1}{2}\geq x$ also for $s<s^{*}\left(x\right)$ by setting $\alpha_{x}=1$. Hence, the sender can always ensure that such voters vote for the government. We can conclude that at the optimum, $\pi_x\left(s\right)\geq x\text{ for all }s\in\mathbb{R}$, as any other outcome implies a lower payoff for the sender.
Therefore,
the optimal strategy solves
\[
\min\alpha_{x}\text{ subject to }\pi_x\left(s\right)\geq x\text{ for all }s\in\mathbb{R}.
\]

We can now determine the optimal shape of $\tilde{f}_{x}$. 
First, we can show that $\pi_x\left(s\right)=x$ for all $s \leq s^{*}\left(x\right)$. To see this, suppose on the contrary
%
that there exists a positive-measure set $B(x)\subseteq\left(-\infty,s^{*}\left(x\right)\right]$ such that 
$\pi_x\left(s\right)>x$ for all $s\in B(x)$. Then the sender can reduce $\tilde{f}_{x}\left(s\right)$ for all $s\in B(x)$ and increase $\tilde{f}_{x}\left(s\right)$ for all $s\in\left(-\infty,s^{*}\left(x\right)\right]\setminus B(x)$ by some small $\epsilon_s>0$. Recall that $\pi_x\left(s\right)$ is increasing in $\alpha_{x}$ and in $\tilde{f}_{x}\left(s\right)$ when $s\leq s^{*}\left(x\right)<0$.
Hence, as a result of this deviation, $\pi_x\left(s\right)>x$ for all $s\in\mathbb{R}$. Then the sender can reduce $\alpha_{x}$ while ensuring that $\pi_x\left(s\right)\geq x$ for all $s\in\mathbb{R}$. Hence, the original strategy of the sender was not optimal.

Next, we show that $\tilde{f}_{x}\left(s\right)=0$ for all $s > s^{*}\left(x\right)$. To see this, suppose the opposite is the case. Let $C(x):=\left\{s>s^{*}\left(x\right):\tilde{f}_{x}\left(s\right)>0 \right\}$.
If $C(x)$ has positive measure, the sender can reduce $\tilde{f}_{x}\left(s\right)$ for $s\in C(x)$ and increase $\tilde{f}_{x}\left(s\right)$ for $s\in\left(-\infty,s^{*}\left(x\right)\right]$. Because $\pi_x\left(s\right) \geq x$ for all $s\in\mathbb{R}$, this deviation ensures that $\pi_x\left(s\right)>x$ for all $s\in\mathbb{R}$. Then the sender can similarly reduce $\alpha_{x}$ while ensuring that $\pi_x\left(s\right)\geq x$ for all $s\in\mathbb{R}$, implying that the original strategy was not optimal.

Let 
\begin{equation}
\label{eq:kappa}
\kappa\left(s\right):=\frac{xF_{0}\left(s\right)-\left(1-x\right)F_{1}\left(s\right)}{1-2x}. 
\end{equation}
The above logic implies that the optimal strategy of the sender is to set $\tilde{f}_{x}\left(s\right)=0$ for all $s> s^{*}\left(x\right)$, and for $s \leq s^{*}\left(x\right)$ to set $\tilde{f}_{x}\left(s\right)$ such that $\pi_x\left(s\right)=x$, that is
\begin{equation}
  \frac{1}{1+\frac{\left(1-\alpha_{x}\right)f_{0}\left(s\right)+\alpha_{x}\tilde{f}_{x}\left(s\right)}{\left(1-\alpha_{x}\right)f_{1}\left(s\right)+\alpha_{x}\tilde{f}_{x}\left(s\right)}}=x
\iff  \tilde{f}_{x}\left(s\right)=\frac{1-\alpha_{x}}{\alpha_{x}}\kappa^{\prime}\left(s\right).\label{eq:f-tilde x<0.5}
\end{equation}

\noindent
Note that $\tilde{f}_{x}\left(s\right)$ as defined above is positive. To see this, observe that at $s=s^{*}\left(x\right)$ we have $m\left(s\right)=\frac{f_{1}\left(s\right)}{f_{0}\left(s\right)}=\frac{x}{1-x}$, and hence $xf_{0}\left(s\right)-\left(1-x\right)f_{1}\left(s\right)=0$. As $m\left(s\right)$ is increasing in $s$, we have $\frac{f_{1}\left(s\right)}{f_{0}\left(s\right)}<\frac{x}{1-x}$ for all $s<s^{*}\left(x\right)$, and hence $xf_{0}\left(s\right)-\left(1-x\right)f_{1}\left(s\right)>0$. 
At the same time, $\alpha_{x}$ must be chosen in such a way that $\tilde{f}_{x}\left(s\right)$ is a density, that is, such that $\int_{-\infty}^{s^{*}\left(x\right)}\tilde{f}_{x}\left(s\right)ds=1$. Consequently, the optimal intensity $\alpha_{x}^*$ solves
\begin{equation}
\label{eq:alpha_x<1/2}
  \int_{-\infty}^{s^{*}\left(x\right)}\dfrac{1-\alpha_{x}^*}{\alpha_{x}^*}\kappa^{\prime}\left(s\right)ds=1
\iff  \alpha_{x}^*=\frac{\kappa\left[s^{*}\left(x\right)\right]}{1+\kappa\left[s^{*}\left(x\right)\right]}.
\end{equation}

\paragraph{Anti-government voters.}
Take a voter with type $x\in\left(\frac{1}{2},1\right]$. For such a voter, $s^{*}\left(x\right)>0$. Given the sender's choice of $\alpha_{x}$ and $\tilde{f}_{x}$, let $A\left(x\right)\subseteq\mathbb{R}$ be the set of signals such that this voter votes for the government if and only if she receives a signal $s\in A\left(x\right)$. That is, $A\left(x\right)$ is a set of signals such that $\pi_x\left(s\right)\geq x$ if and only if $s\in A\left(x\right)$. Then the probability that the voter votes for the government in state $\theta\in\left\{ 0,1\right\} $ equals the probability that she observes a signal $s\in A\left(x\right)$. This probability equals
\[
p_\theta\left(x\right)=\int_{s\in A\left(x\right)}\left[\left(1-\alpha_{x}\right)f_{\theta}\left(s\right)+\alpha_{x}\tilde{f}_{x}\left(s\right)\right]ds.
\]

Next we state a lemma that defines the optimal $A(x)$:

\begin{lem} \label{lem:A(x)_IF_x>1/2}
At the optimum, $A(x)=\left[s^*(x),\infty\right)$ for all $x>\ot$.
\end{lem}

\begin{proof}
First, we show that if $s<s^{*}\left(x\right)$, then $s\notin A\left(x\right)$.
Recall that $\pi_x\left(s\right)$ is decreasing in $\alpha_{x}$ and in $\tilde{f}_{x}\left(s\right)$  if and only if $s>0$.
Consider any signal $s$. If $s<0$, then 
$\pi_x\left(s\right)<\frac{1}{2}<x$ regardless of $\alpha_{x}$ and $\tilde{f}_{x}$. Hence, if $s<0$, then $s\notin A\left(x\right)$. On the other hand, if $s\in\left[0,s^{*}\left(x\right)\right)$, then without trolls (that is, when $\alpha_{x}=0$ or $\tilde{f}_{x}\left(s\right)=0$), the voter votes against the government. Since $\pi_x\left(s\right)$ is decreasing in $\alpha_{x}$ and in $\tilde{f}_{x}\left(s\right)$, the voter votes against the government for any $\alpha_{x}$ and $\tilde{f}_{x}$. Hence, if $s\in\left[0,s^{*}\left(x\right)\right)$, then $s\notin A\left(x\right)$.

To show that all $s \geq s^{*}\left(x\right)$ belong to $A\left(x\right)$, assume to the contrary that there exists a set $B(x)\subseteq \left[s^*(x),\infty\right)$ with positive measure  such that $\pi_x(s)<x$ for all $s\in B(x)$. Note that 
it must be true that $\alpha_x>0$ \emph{and} $\tilde{f}_x(s)>0$ all $s\in B(x)$, because otherwise $s\in A(x)$. Now take an arbitrary  strict subset $C(x)\subset B(x)$ with positive measure, let $\tilde{f}_x(s)=0$ for all $s\in C(x)$,  increase $\tilde{f}_x(s)$ for all $s\in B(x)\setminus C(x)$ such that $\tilde{f}_x(s)$ remains a density, while leaving $\tilde{f}_x(s)$ unchanged for all $s\in A(x)\setminus B(x)$. As a result, $C(x)$ becomes part of $A\left(x\right)$.
This strictly increases
 the sender's payoff, which implies that the original strategy was not optimal. We can thus conclude that, at the optimum, $A\left(x\right)=\left[s^{*}\left(x\right),\infty\right)$.
\end{proof}
%
%


Next, we can show that it is optimal for the sender to set $\pi_x\left(s\right)=x$ for all $s\geq s^{*}\left(x\right)$. 
To see this, take a set $D\left(x\right)\subset\left[s^{*}\left(x\right),\infty\right)$, and suppose that $\pi_x\left(s\right)>x$ for all $s\in D(x)$. If there also exists a set $D'\left(x\right)\subset\left(-\infty,s^{*}\left(x\right)\right]$ such that $\tilde{f}_{x}\left(s\right)>0$ for all $s\in D'\left(x\right)$, then the sender can reduce $\tilde{f}_{x}\left(s\right)$ for $s\in D'\left(x\right)$, and increase $\tilde{f}_{x}\left(s\right)$ for $s\in D\left(x\right)$ such that $\pi_x\left(s\right)$ remains above $x$ for $s\in D(x)$. This leaves $A\left(x\right)$ unchanged while increasing $p_\theta\left(x\right)$ in both states, implying that the original strategy is not optimal. On the other hand, if there is no such a set $D'\left(x\right)$, then $p_{\theta}\left(x\right)=\left(1-\alpha_{x}\right)\left(1-F_{\theta}\left[s^{*}\left(x\right)\right]\right)+\alpha_{x}$, which is increasing in $\alpha_{x}$. Then the sender can increase $\tilde{f}_{x}\left(s\right)$ for all $s\in D(x)$ while decreasing $\tilde{f}_{x}\left(s\right)$ for all $s\in\left[s^{*}\left(x\right),\infty\right)\setminus D(x)$ in such a way that $\pi_x\left(s\right)>x$ for all $s\geq s^{*}\left(x\right)$ as a result. Then the sender can increase $\alpha_{x}$ while ensuring that $\pi_x\left(s\right)\geq x$ still holds for all $s\geq s^{*}\left(x\right)$. This would increase $p_{\theta}\left(x\right)$, implying that the original strategy was not optimal.

%
%
Hence, at the optimum 
for all $s\geq s^{*}\left(x\right)$, $\tilde{f}_{x}\left(s\right)$ is such that
\begin{align}
 & \pi_x(s)=\frac{1}{1+\frac{\left(1-\alpha_{x}\right)f_{0}\left(s\right)+\alpha_{x}\tilde{f}_{x}\left(s\right)}{\left(1-\alpha_{x}\right)f_{1}\left(s\right)+\alpha_{x}\tilde{f}_{x}\left(s\right)}}=x\nonumber \\
\iff & \tilde{f}_{x}\left(s\right)=\frac{1-\alpha_{x}}{\alpha_{x}}\frac{\left(1-x\right)f_{1}\left(s\right)-xf_{0}\left(s\right)}{2x-1}=\dfrac{1-\alpha_{x}}{\alpha_{x}}\kappa^{\prime}\left(s\right).\label{eq:f-tilde}
\end{align}
Note that $\tilde{f}_{x}\left(s\right)$ as defined above is strictly positive. To see this, observe that at $s=s^{*}\left(x\right)$ we have $m\left(s\right)=\frac{f_{1}\left(s\right)}{f_{0}\left(s\right)}=\frac{x}{1-x}$, and hence $\left(1-x\right)f_{1}\left(s\right)-xf_{0}\left(s\right)=0$. As $m\left(s\right)$ is increasing in $s$, we have $\frac{f_{1}\left(s\right)}{f_{0}\left(s\right)}>\frac{x}{1-x}$ for all $s>s^{*}\left(x\right)$, and hence $\left(1-x\right)f_{1}\left(s\right)-xf_{0}\left(s\right)>0$.

Then the probability that the voter votes for the government in state $\theta$ equals
\[
p_\theta\left(x\right)=\left(1-\alpha_{x}\right)\int_{s\in A\left(x\right)}\left[f_{\theta}\left(s\right)+\kappa^{\prime}\left(s\right)\right]ds.
\]

Thus, conditional on $\tilde{f}_{x}\left(s\right)$ being such that $\pi_x\left(s\right)=x$ for all $s\geq s^{*}\left(x\right)$, the vote share of the government in each state is strictly decreasing in $\alpha_x$. Hence, it is optimal to select the smallest possible $\alpha_{x}$ under which $\tilde{f}_{x}\left(s\right)$ is a density and $\pi_x\left(s\right)=x$ for all $s\geq s^{*}\left(x\right)$. This implies that $\tilde{f}_{x}\left(s\right)=0$ for all $s \notin A\left(x\right)$, because otherwise, the sender can decrease $\tilde{f}_{x}\left(s\right)$ for $s \notin A\left(x\right)$, increase it for all $s \in A\left(x\right)$, and reduce $\alpha_{x}$ such that the belief goes up to $x$ for all $s \in A\left(x\right)$. This increases $p_\theta\left(x\right)$, implying that the original strategy was not optimal.

Therefore, the sender sets $\alpha_{x}$ such that $\tilde{f}_{x}\left(s\right)=0$ for all $s<s^{*}\left(x\right)$, and $\int_{s^{*}\left(x\right)}^{\infty}\tilde{f}_{x}\left(s\right)ds=1$. Consequently, $\alpha_{x}^*$ is given by

\begin{equation}
\label{eq:alpha_x>1/2}
\int_{s^{*}\left(x\right)}^{\infty}\frac{1-\alpha_{x}^*}{\alpha_{x}^*}\kappa^{\prime}\left(s\right)ds=1
\iff  \alpha_{x}^*=\frac{\kappa\left[s^{*}\left(x\right)\right]+1}{\kappa\left[s^{*}\left(x\right)\right]}.
\end{equation}

To prove the statement about the vote shares, note that each voter with type $x\leq\frac{1}{2}$ votes for the government with probability one. The mass of such voters is $H\left(\frac{1}{2}\right)$. A voter with type $x\in\left(\frac{1}{2},1\right]$ votes for the government if and only if she receives a signal $s\ge s^{*}\left(x\right)$. In state $\theta\in\left\{ 0,1\right\} $, the probability of this is
\begin{equation}\label{eq:vote share x}
V_{\theta}(x):=\alpha_{x}^{*}+\left(1-\alpha_{x}^{*}\right)\left(1-F_{\theta}\left[s^{*}\left(x\right)\right]\right)=\frac{xF_{0}\left(s\right)-\left(1-x\right)F_{1}\left(s\right)-\left(2x-1\right)F_{\theta}\left[s^{*}\left(x\right)\right]}{xF_{0}\left(s\right)-\left(1-x\right)F_{1}\left(s\right)}.
\end{equation}
Hence, the overall vote share of the government in state $\theta$ is
\begin{equation}\label{eq:vote share}
V_\theta =H\left(\frac{1}{2}\right)+\int_{\frac{1}{2}}^{1}V_\theta(x)dH\left(x\right).    
\end{equation}
Note that $F_{0}\left(s\right)>F_{1}\left(s\right)$, because MLRP implies first-order stochastic dominance. Hence, the denominator of \eqref{eq:vote share x} is positive for all $x>\frac{1}{2}$, and  thus $V_{1}>V_{0}\geq H\left(\frac{1}{2}\right)$.

The result that $V_{\theta}>V^{NT}_{\theta}$ for each $\theta$ follows from the fact that without trolls, in state $\theta$ a voter with type $x\in \left(0,1\right)$ votes for the government with probability $1-F_{\theta}\left[s^{*}\left(x\right)\right]$. With trolls, a pro-government voter votes for the government with probability 1, while an anti-government voter votes for the government with probability $V_{\theta}(x)>1-F_{\theta}\left[s^{*}\left(x\right)\right]$.
\qed

\subsection{Proof of Corollary \ref{cor:dadx}}
The derivative of $\alpha_x^*$ for  $x<\ot$ as defined in \eqref{eq:alpha_x<1/2} is:
\[
\left.\dfrac{\partial \alpha_x^*}{\partial x}\right|_{x<\ot}
=\frac{-(2 x-1) \frac{ds^*(x)}{dx} \left[x {f_0}(s^*(x))-(1-x)
   {f_1}(s^*(x))\right]+{F_0}(s^*(x))-{F_1}(s^*(x))}{(x
   {F_0}(s^*(x))+(x-1) {F_1}(s^*(x))-2 x+1)^2}.
\]
First note that $m(s^*(x))=\frac{x}{1-x}$, which implies $\left[x {f_0}(s^*(x))-(1-x)
   {f_1}(s^*(x))\right]=0$. Moreover, the denominator of the derivative is strictly positive. Hence, $\left.\frac{\partial \alpha_x^*}{\partial x}\right|_{x<\ot}>0\iff {F_0}(s^*(x))-{F_1}(s^*(x))>0$, which is true because of the  MLRP. This proves this part of the result.

The derivative of $\alpha_x^*$  for  $x>\ot$ as defined in \eqref{eq:alpha_x>1/2} is:
\[
\left.\dfrac{\partial \alpha_x^*}{\partial x}\right|_{x>\ot}
=\frac{(2 x-1) \frac{ds^*(x)}{dx} \left[x {f_0}(s^*(x))-(1-x)
   {f_1}(s^*(x))\right]+{F_1}(s^*(x))-{F_0}(s^*(x))}{(x
   {F_0}(s^*(x))+(x-1) {F_1}(s^*(x)))^2}.
\]
The expression in brackets is again equal to zero, and the denominator again strictly positive. 
 Hence, $\left.\frac{\partial \alpha_x^*}{\partial x}\right|_{x>\ot}<0\iff {F_1}(s^*(x))-{F_0}(s^*(x))<0$, which again follows from the MLRP. This proves the second part of the corollary.\qed

\subsection{Proof of Theorem \ref{prop: informativeness_Blackwell}}

Using \eqref{eq:vote share x}  and \eqref{eq:vote share}, we obtain
\begin{equation}\label{eq:V0}
V_{0}=H\left(\frac{1}{2}\right)+\int_{\frac{1}{2}}^{1}\frac{\frac{F_{0}\left[s^{*}\left(x\right)\right]}{F_{1}\left[s^{*}\left(x\right)\right]}-1}{\frac{x}{1-x}\frac{F_{0}\left[s^{*}\left(x\right)\right]}{F_{1}\left[s^{*}\left(x\right)\right]}-1}dH\left(x\right),
\end{equation}
and
\begin{equation}\label{eq:V1}
V_{1}=H\left(\frac{1}{2}\right)+\int_{\frac{1}{2}}^{1}\frac{\frac{F_{0}\left[s^{*}\left(x\right)\right]}{F_{1}\left[s^{*}\left(x\right)\right]}-1}{\frac{F_{0}\left[s^{*}\left(x\right)\right]}{F_{1}\left[s^{*}\left(x\right)\right]}-\frac{1-x}{x}}dH\left(x\right).
\end{equation}
An increase in informativeness implies that $\frac{F_{0}\left[s^{*}\left(x\right)\right]}{F_{1}\left[s^{*}\left(x\right)\right]}$ increases for all $x\in\left[\frac{1}{2},1\right]$. Therefore,  both $V_{0}$ and $V_{1}$ increase.\qed

\subsection{Proof of Theorem \ref{prop: polarisation}}
Let 
\[
\begin{array}{ccc}
\tilde{V}_{\theta}(x):=\begin{cases}
1 & \text{ if }x\leq\frac{1}{2}\\
V_\theta(x) & \text{ if }x>\frac{1}{2}
\end{cases}&\text{ and } &
\tau\left[H\left(x\right)\right]:=\begin{cases}
H\left(\frac{1}{2}\right) & \text{ if }x\leq\frac{1}{2}\\
H\left(x\right) & \text{ if }x>\frac{1}{2}
\end{cases},
\end{array}
\]
where $V_\theta(x)$ is defined by \eqref{eq:vote share x}. 
Using \eqref{eq:vote share} we have $V_{\theta}=\int_{0}^{1}\tilde{V}_{\theta}(x)d\tau\left[H\left(x\right)\right].$ By Definition \ref{def: polarisation}, $\tau[\hat{H}]$ first-order stochastically dominates $\tau[H]$. Furthermore, rewriting \eqref{eq:vote share x}, we have $V_\theta(x)=1-\left(1-\alpha_{x}^{*}\right)F_{\theta}\left[s^{*}\left(x\right)\right]$. By Corollary \ref{cor:dadx}, $\alpha_{x}^{*}$ is decreasing in $x$, while $s^{*}\left(x\right)$ is increasing in $x$. Therefore, $V_\theta(x)$, and hence also $\tilde{V}_{x}^{\theta}$, is decreasing in $x$. Stochastic dominance then implies that $V_{\theta}$ is smaller under $\hat{H}$ than under $H$ for $\theta\in\left\{ 0,1\right\} $. \qed

\subsection{Proof of Proposition \ref{prop: informativeness}}

Recall that $V_0$ and $V_1$ are defined by \eqref{eq:V0} and  \eqref{eq:V1}, and $V_{0}<V_{1}$. Also, by Theorem \ref{prop: informativeness_Blackwell},  $V_{0}$ and $V_{1}$ are increasing with $r$.
We have $
\lim_{r\rightarrow-\infty}V_{0}=\lim_{r\rightarrow-\infty}V_{1}=H\left(\frac{1}{2}\right)<\frac{1}{2}$,
hence the government loses the election in both states when $r$ is sufficiently small. Furthermore,
\[
\lim_{r\rightarrow\infty}V_{1}=H\left(\frac{1}{2}\right)+\int_{\frac{1}{2}}^{1}dH\left(x\right)=1,
\]
hence the government wins the election in state 1 when $r$ is sufficiently high. By monotonicity, there exists $r^{\prime}$ such that the government loses the election state 1 if and only if $r<r^{\prime}$. At the same time, 
\begin{align*}
\lim_{r\rightarrow\infty}V_{0}= & H\left(\frac{1}{2}\right)+\int_{\frac{1}{2}}^{1}\frac{1-x}{x}dH\left(x\right)\\
= & H\left(\frac{1}{2}\right)+\left[\frac{1-x}{x}H\left(x\right)\right]_{\frac{1}{2}}^{1}+\int_{\frac{1}{2}}^{1}\frac{1}{x^{2}}H\left(x\right)dx
=  \int_{\frac{1}{2}}^{1}\frac{1}{x^{2}}H\left(x\right)dx.
\end{align*}
If polarisation is low enough, then $H\left(x\right)$ is close to 1 for almost all $x\in\left(\frac{1}{2},1\right)$. Hence, $\lim_{r\rightarrow\infty}V_{0}$ is close to $\int_{\frac{1}{2}}^{1}\frac{1}{x^{2}}dx=1>\frac{1}{2}$. Therefore, the government wins the election in state 0 when $r$ is sufficiently high. By monotonicity, there exists $r^{\prime\prime}$ such that the government loses the election in state 0 for $r<r^{\prime\prime}$, and wins the election in state 0 if and only if $r>r^{\prime\prime}$. The fact that $V_{0}<V_{1}$ implies that  $r^{\prime}<r^{\prime\prime}$. Hence, when  $r\in\left(r^{\prime},r^{\prime\prime}\right)$, the government wins the election only in state 1.

If polarisation is high enough, $H\left(x\right)$ is close to $H\left(\frac{1}{2}\right)$ for almost all $x\in\left(\frac{1}{2},1\right)$. Hence, $\lim_{r\rightarrow\infty}V_{0}$ is close to $H\left(\frac{1}{2}\right)\int_{\frac{1}{2}}^{1}\frac{1}{x^{2}}dx=H\left(\frac{1}{2}\right)<\frac{1}{2}$. Therefore, the government loses the election in state 0 for all $r$. 
By  Theorem \ref{prop: polarisation}, a change from $H$ to $\hat{H}$ which dominates $H$ in polarisation decreases $V_{0}$. Thus, for  $r>r^{\prime\prime}$, the government wins the election in both states if polarisation is sufficiently low, and wins the election in state 1 only otherwise.\qed

\subsection{Proof of Corollary \ref{cor:trolls-efficiency}}

By Proposition \ref{prop: informativeness},  there exists an information structure such that $V_1=\frac{1}{2}>V_0$. By Proposition \ref{lem: strategy}, without trolls under this information structure, $\frac{1}{2}>V^{NT}_1>V^{NT}_0$.\qed

\subsection{Proof of Proposition \ref{prop:costly trolls}}
To prove the proposition we proceed in two steps. We first derive the optimal slant $\tilde{f}_x(s)$ of the trolls' messages for a given intensity $\alpha_x$. 
We then use these results to find the optimal intensity $\bar{\alpha}_x^*$ when trolls are costly.

\subsubsection{Optimal Slant}

Recall from \eqref{eq:belief} that after observing a signal $s$, the voter forms a posterior belief 
\[
\pi_x\left(s\right)=\frac{1}{1+\frac{\left(1-\alpha_{x}\right)f_{0}\left(s\right)+\alpha_{x}\tilde{f}_{x}\left(s\right)}{\left(1-\alpha_{x}\right)f_{1}\left(s\right)+\alpha_{x}\tilde{f}_{x}\left(s\right)}}.
\]
As in the proof of Proposition \ref{lem: strategy}, for a voter with type $x$, let $A\left(x\right)\subseteq\mathbb{R}$ be the set of signals such that this voter votes for the government if and only if she receives a signal $s\in A\left(x\right)$. Therefore, $\pi_x(s)\geq x$ if and only if $s\in A\left(x\right)$.

\paragraph{Anti-government voters.}
First, take a voter with type $x\in\left(\frac{1}{2},1\right]$. It follows again from   Lemma \ref{lem:A(x)_IF_x>1/2} that at the optimum  $A\left(x\right)=\left[s^{*}\left(x\right),\infty\right)$. 



Let $p_\theta\left(x\right)$ be the probability that the voter votes for the government in state $\theta\in\left\{ 0,1\right\} $. It equals  the probability that she observes a signal $s\in A\left(x\right)$. Thus,
\[
\begin{array}{rcl}
p_\theta\left(x\right)&=&\int_{s\in A\left(x\right)}\left[\left(1-\alpha_{x}\right)f_{\theta}\left(s\right)+\alpha_{x}\tilde{f}_{x}\left(s\right)\right]ds\\
&=&(1-\alpha_x)\left(1-F_\theta\left[s^*(x)\right]\right)+\alpha_{x}\left(1-\tilde{F}_{x}\left[s^*(x)\right]\right).
\end{array}
\]
This is decreasing in $\tilde{F}_{x}\left[s^*(x)\right]$. Hence, at the optimum, $\tilde{F}_{x}\left[s^*(x)\right]=0$.

\paragraph{Pro-government voters.}
Now take a voter with type $x\in\left[0,\frac{1}{2}\right]$. Recall that the optimal intensity $\alpha_x^*$ as defined in \eqref{eq:alpha_x<1/2} minimizes $\alpha_{x}$ subject to the constraint that for these voters $A\left(x\right)=\mathbb{R}$, i.e., that every pro-government voter votes for the government. If $\alpha_{x}\geq \alpha_{x}^*$, then the sender can likewise induce every voter to vote for the government by setting the support of $\tilde{f}_{x}\left(s\right)$ as in Proposition \ref{lem: strategy}.
If $\alpha_{x}<\alpha_{x}^*$, the constraint must be violated, so $A\left(x\right)\subset\mathbb{R}$. For the rest of the section we focus on this case.

%

We now show that at the optimum $\tilde{f}_{x}\left(s\right)=0$ for all $s \in  \left[s^{*}\left(x\right),\infty\right) \bigcup \left( \mathbb{R}\setminus A\left(x\right)\right)$, except possibly for a zero-measure set. If there was a positive-measure set 
\[C\left(x\right)\subseteq \left[s^{*}\left(x\right),\infty\right) \bigcup \left( \mathbb{R}\setminus A\left(x\right)\right)\] such that $\tilde{f}_{x}\left(s\right)>0$ for all $s\in C\left(x\right)$, the sender could move the probability mass from $C\left(x\right)$ to some interval outside $A\left(x\right)$ so that the belief in that interval increased to $x$. This would expand $A\left(x\right)$, increasing the share of voters voting for the government. Following the same reasoning, for all $s<s^{*}\left(x\right)$ we cannot have $\pi_x\left(s\right)=x$. Hence, for all $s<s^{*}\left(x\right)$ such that $s\in A\left(x\right)$ we have
\[
\pi_x\left(s\right)=x \iff \tilde{f}_{x}\left(s\right) = \frac{1-\alpha_{x}}{\alpha_x}\kappa^{\prime}\left(s\right), 
\]
where the last expression follows \eqref{eq:f-tilde x<0.5} and $\kappa(s)$ is  defined in \eqref{eq:kappa}.

We  now determine the support  of $\tilde{f}_x(s)$, which we denote by $\mathcal{S}\left(x\right)$. The preceding analysis implies that $\mathcal{S}\left(x\right)\subset A\left(x\right)$, and $\mathcal{S}\left(x\right)\subset \left(-\infty,s^{*}\left(x\right)\right]$. 
Take a signal $s$ in the interior of $\mathcal{S}\left(x\right)$, and a $\delta>0$ such that the interval $D:=\left(s-\delta,s+\delta\right)\in \mathcal{S}(x)$.
When $\delta$ is sufficiently small, the mass of trolls' signals in the interval $D$ approximately equals
\[
2\delta\tilde{f}_{x}\left(s\right)=2\delta \frac{1-\alpha_{x}}{\alpha_{x}} \kappa^{\prime}\left(s\right)=2\delta\frac{1-\alpha_{x}}{\alpha_{x}}\frac{xf_{0}\left(s\right)-\left(1-x\right)f_{1}\left(s\right)}{1-2x}.
\]
Now consider the following deviation: remove $D$ from $\mathcal{S}(x)$, and instead move the mass of trolls' signals to an interval around some signal $s^{\prime}\notin \mathcal{S} \left(x\right)$ such that $s^{\prime}<s^*\left(x\right)$, while making sure that $\pi_x\left(s^{\prime}\right)=x$. The latter condition means that $\tilde{f}_{x}\left(s'\right)=\frac{1-\alpha_{x}}{\alpha_{x}}\frac{xf_{0}\left(s'\right)-\left(1-x\right)f_{1}\left(s'\right)}{1-2x}$. The
width of the interval around $s^{\prime}$ that is added to $\mathcal{S}(x)$ as a consequence of this deviation equals approximately
\[
2\delta\frac{\tilde{f}_{x}\left(s\right)}{\tilde{f}_{x}\left(s^{\prime}\right)}=2\delta\frac{xf_{0}\left(s\right)-\left(1-x\right)f_{1}\left(s\right)}{xf_{0}\left(s^{\prime}\right)-\left(1-x\right)f_{1}\left(s^{\prime}\right)}.
\]
In  state $\theta\in\left\{0,1\right\}$, this deviation leads to a change in the probability that a voter of type $x$ votes for the government that approximately equals
\[
\left(1-\alpha_{x}\right)2\delta\frac{xf_{0}\left(s\right)-\left(1-x\right)f_{1}\left(s\right)}{xf_{0}\left(s^{\prime}\right)-\left(1-x\right)f_{1}\left(s^{\prime}\right)}f_{\theta}\left(s^{\prime}\right)-\left(1-\alpha_{x}\right)2\delta f_{\theta}\left(s\right).
\]
The deviation is beneficial if and only if this is positive, i.e., if and only if
\[
\dfrac{xf_{0}\left(s\right)-\left(1-x\right)f_{1}\left(s\right)}{xf_{0}\left(s^{\prime}\right)-\left(1-x\right)f_{1}\left(s^{\prime}\right)}>\dfrac{f_{\theta}\left(s\right)}{f_{\theta}\left(s^{\prime}\right)}.
\]
\noindent
In state $\theta=0$ this is satisfied if and only if
\[
  \frac{xf_{0}\left(s\right)-\left(1-x\right)f_{1}\left(s\right)}{xf_{0}\left(s^{\prime}\right)-\left(1-x\right)f_{1}\left(s^{\prime}\right)}>\frac{f_{0}\left(s\right)}{f_{0}\left(s^{\prime}\right)}
\iff \frac{f_{1}\left(s^{\prime}\right)}{f_{0}\left(s^{\prime}\right)}> \frac{f_{1}\left(s\right)}{f_{0}\left(s\right)}
\iff  s'>s,
\]
\noindent
while in state $\theta=1$ it has to hold that
\[
  \frac{xf_{0}\left(s\right)-\left(1-x\right)f_{1}\left(s\right)}{xf_{0}\left(s^{\prime}\right)-\left(1-x\right)f_{1}\left(s^{\prime}\right)}>\frac{f_{1}\left(s\right)}{f_{1}\left(s^{\prime}\right)}
\iff  \frac{f_{1}\left(s^{\prime}\right)}{f_{0}\left(s^{\prime}\right)}>\frac{f_{1}\left(s\right)}{f_{0}\left(s\right)}
\iff  s^{\prime}>s.
\]
\noindent
Hence, in each state $\theta$ a deviation is profitable if and only if there is some $s'<s^*\left(x\right)$ that is larger than some $s\in \mathcal{S}(x)$. This in turn implies that there exists no beneficial deviation if and only if $\mathcal{S}(x)=\left[\hat{s}\left(x\right),s^{*}\left(x\right)\right]$ for some $\hat{s}\left(x\right)<s^{*}\left(x\right)$, 
where $\hat{s}\left(x\right)$ solves
$\int_{\hat{s}\left(x\right)}^{s^{*}\left(x\right)}\frac{1-\alpha_{x}}{\alpha_{x}}\kappa^{\prime}\left(s\right)ds=1$.
Further, to see that $\hat{s}\left(x\right)$ always exists and is unique, note that when $\hat{s}\left(x\right)\rightarrow s^*(x)$, the integral approaches zero. At the same time, as $\hat{s}\left(x\right)\rightarrow -\infty$, the integral becomes strictly greater than one because $\alpha_x<\alpha_x^*$ and the integral equals one when $\alpha_x=\alpha_x^*$. Existence of  $\hat{s}(x)$ then follows because the integral is continuous in $\hat{s}(x)$. Uniqueness follows from the fact that $\kappa^{\prime}\left(s\right)>0$ for all $s<s^*\left(x\right)$, and hence the integral is strictly monotone in $\hat{s}\left(x\right)$.

\subsubsection{Optimal Intensity}


\paragraph{Anti-Government Voters:}
We know from the analysis of optimal slant that voters with $x\geq\ot$ vote for the government if and only if they receive a signal $s \geq s^*(x)$.
Thus, the government's  vote share among these voters   in state $\theta$ is
$V_\theta(x)=\alpha_x +(1-\alpha_x) \left(1-F_\theta \left(s^*(x)\right)\right).$
The sender's expected utility is then 
\[\frac{1}{2}\left(\int_{0}^1V_0(x) dH(x) +\int_{0}^1V_1(x) dH(x)\right) -h(x)\,c\left(\alpha_x\right).
\]
The first derivative of expected utility with respect to $\alpha_x$ is
\begin{equation}
\frac{1}{2} h(x)\left(F_0\left(s^*(x)\right)+F_1\left(s^*(x)\right)\right)-h(x)c'\left(\alpha_x\right).
\label{eq:FOC_cost_AG}  
\end{equation}
In an interior equilibrium, this has to equal zero. Note that 
$h(x)$ does not affect the optimal strategy. 
Because the optimal intensity cannot exceed $\alpha^*_x$ as defined \eqref{eq:alpha_x>1/2}, the optimal $\alpha_x$ can be of three kinds: (i) zero, (ii) $\alpha^*_x$, or (iii) interior $\alpha_x\in(0,\alpha^*_x)$.
In an interior solution, we have
\[
\bar{\alpha}_x^*=c'^{-1}\left(\frac{1}{2}F_0(s^*(x))+F_1(s^*(x))\right).
\]
This solution constitutes the unique optimum if it belongs to the $\left(0,\alpha_x^*\right)$ interval. 
Otherwise, we are in one of the two corner solutions.

Note that because 
\begin{equation}
\frac{1}{2} \left(f_0\left(s^*(x)\right)+f_1\left(s^*(x)\right)\right)\frac{\partial s^*(x)}{\partial x}>0,
\label{eq:MBEN}
\end{equation}
the marginal benefit of troll is increasing in $x$ for all $x>\ot$. Moreover, note that $\alpha_{x=1}^*=0$. 
Hence, if $\bar{\alpha}_
{x'}^*>0$ for some $x'>\frac{1}{2}$, then for every $x\in(x',1)$ also $\bar{\alpha}_x^*>0$. Moreover, if $\bar{\alpha}_
{x''}^*=0$ for some $x''>\frac{1}{2}$, then for every $x\in(\frac{1}{2},x'')$ also $\bar{\alpha}_x^*=0$. Hence, we can conclude that there exists $\hat{x}\in[\frac{1}{2},1]$ such that $\bar{\alpha}_x^*>0$ iff $x\in[\hat{x},1)$.

Moreover, note that if $c'$ is low enough  but positive for all $\alpha_x\in[0,\alpha_x^*]$ and $x\geq\ot$, then \eqref{eq:FOC_cost_AG} is positive for all $\alpha_x\in[0,\alpha_x^*]$. It follows that in this case  we have $\bar{\alpha}_x^*=\alpha_x^*$ for all $x\geq\ot$.

\paragraph{Pro-Government Voters:}
Next consider pro-government voters with $x\leq \ot$. We know from the analysis of optimal slant that for any $\alpha_x$ below $\alpha^*_x$ as defined in \eqref{eq:alpha_x<1/2}, the optimal strategy is to emulate only signals $s\in[\hat{s}(x,\alpha_x),s^*(x)]$, where $\hat{s}\left(x,\alpha_x\right)$ solves
$\int_{\hat{s}(x,\alpha_x)}^{s^{*}\left(x\right)}\kappa^{\prime}\left(s\right)ds=\frac{\alpha_x}{1-\alpha_x}$. The derivative of this condition with respect to $\alpha_x$ is $-\frac{\partial \hat{s}(x,\alpha_x)}{\partial \alpha_x}\kappa'(\hat{s}(x,\alpha_x))=\frac{1}{(1-\alpha_x)^2},$
and hence, 
\begin{equation}
\label{eq:dsda}
\frac{\partial \hat{s}(x,\alpha_x)}{\partial \alpha_x}=-\frac{1}{\kappa'(\hat{s}(x,\alpha_x))(1-\alpha_x)^2}<0.
\end{equation}
Greater  $\alpha_x$ means  a greater set of signals is emulated by the sender. This will be important to determine the optimal $\alpha_x$ for pro-government voters.

Given optimal $\tilde{f}_x(s)$, each voter receiving a signal $s$ that \textit{could} have been sent by a troll votes for the government. Therefore, the share of voters of type $x$ voting for the government in  state $\theta$ is
\[
V_\theta(x)=\alpha_x+(1-\alpha_x)(1-F_\theta(\hat{s}(x,\alpha_x))).
\]
Expected utility is defined in analogy to the case of anti-government voters.
Then the first derivative with respect to $\alpha_x$ is
\[
\frac{1}{2} h(x)\left(F_0(\hat{s}(x,\alpha_x ))+F_1(\hat{s}(x,\alpha_x ))-(1-\alpha_x)
   \frac{\partial \hat{s}(x,\alpha_x)}{\partial \alpha_x} \left(f_1(\hat{s}(x,\alpha_x
   ))+f_0(\hat{s}(x,\alpha_x
   ))\right)\right)-h(x)c_x'(\alpha_x).
\]
Again, $h(x)$ does not affect the  sender's optimal strategy. Using \eqref{eq:dsda} and ignoring $h(x)$, the marginal benefit of trolls becomes
\[
\frac{1}{2} \left(F_0(\hat{s}(x,\alpha_x))+F_1(\hat{s}(x,\alpha_x))+\frac{f_0(\hat{s}(x,\alpha_x))+f_1(\hat{s}(x,\alpha_x))}{(1-\alpha_x) \kappa '(\hat{s}(x,\alpha_x))}\right).
\]
This is strictly positive for all $\alpha_x\leq \alpha_x^*$. Recall that
\[
\kappa'(s)=\frac{xf_{0}\left(s\right)-\left(1-x\right)f_{1}\left(s\right)}{1-2x},
\]
and thus $\left.\kappa'(s)\right|_{s<s^*(x)\wedge x<\frac{1}{2}}>0$, and $\left.\kappa'(s)\right|_{s=s^*(x)\wedge x<\frac{1}{2}}=0$. Hence, as  $\alpha_x$ approaches zero, $\lim_{\alpha_x\rightarrow 0^+}\kappa'(\hat{s}(x,\alpha_x))=0^+$. Therefore, the marginal benefit of trolls becomes infinite as $\alpha_x\rightarrow 0^+$.  Consequently, any $x<\ot$ will be targeted by the sender with $\bar{\alpha}_x^*>0$ because $c_x'(0)$ is finite. Because it is unclear how  the marginal benefit of trolls behaves generally, the exact optimum is not clear without further assumptions. 

Moreover, note that for low enough but positive $c'$ for all $\alpha_x$,  for all $x\leq\ot$  we have $\bar{\alpha}_x^*=\alpha_x^*$. This follows from the marginal benefit being \textit{strictly} positive for all $\alpha_x\in[0,\alpha_x^*]$.\qed

\subsection{Proof of Lemma \ref{lem:garblings}}

Take any voter type $x$. Each information structure  $\left(F,G\right)$ induces a probability that this voter votes for the government. Let $M_{x}$ be the set of these probabilities under different information structures $\left(F,G\right)$. 
Since the set of states is binary, it follows from \cite{KamenicaGentzkow:2011} that the largest and the smallest elements of $M_{x}$ can also obtained under binary signals. Convexity of $M_{x}$ then implies that any payoff that can be obtained under an arbitrary $\left(F,G\right)$ can also be obtained under a binary signal. 
Hence, it is sufficient to prove the result for binary signals $s$.


Pick any voter, and take any $G$. Let $s\in\left\{ 0,1\right\} $, and let $p_{\theta}$ be the probability that $s=1$ when the state is $\theta$. Take a garbling of signal $s$ that transforms it to a different binary signal $s^{\prime}\in\left\{ 0,1\right\} $. Let $q_{s}:=\Pr\left(s^{\prime}=1\mid s\right)$ be the probability that, under this garbling, signal $s\in\left\{ 0,1\right\} $ is transformed into signal $s^{\prime}=1$.

Take a given realisation of the non-replaceable signal $\rho$. Under garbling, let $\pi^{G}\left(s^{\prime}=1,\rho\right)$ and $\pi^{G}\left(s^{\prime}=0,\rho\right)$ be the voter's posterior beliefs when she observers, respectively, signals $s^{\prime}=1$ and $s^{\prime}=0$, as well as $\rho$. These beliefs are:
\[\begin{array}{c}
\pi^{G}\left(s^{\prime}=1,\rho\right)=\frac{p_{1}q_{1}g_{1}\left(\rho\right)+\left(1-p_{1}\right)q_{0}g_{1}\left(\rho\right)}{p_{1}q_{1}g_{1}\left(\rho\right)+\left(1-p_{1}\right)q_{0}g_{1}\left(\rho\right)+p_{0}q_{1}g_{0}\left(\rho\right)+\left(1-p_{0}\right)q_{0}g_{0}\left(\rho\right)}
=\dfrac{\frac{p_{1}q_{1}+\left(1-p_{1}\right)q_{0}}{p_{0}q_{1}+\left(1-p_{0}\right)q_{0}}\frac{g_{1}\left(\rho\right)}{g_{0}\left(\rho\right)}}{\frac{p_{1}q_{1}+\left(1-p_{1}\right)q_{0}}{p_{0}q_{1}+\left(1-p_{0}\right)q_{0}}\frac{g_{1}\left(\rho\right)}{g_{0}\left(\rho\right)}+1},
\end{array}
\]
and 
\[
\begin{array}{c}
\pi^{G}\left(s^{\prime}=0,\rho\right)=\frac{p_{1}\left(1-q_{1}\right)g_{1}\left(\rho\right)+\left(1-p_{1}\right)\left(1-q_{0}\right)g_{1}\left(\rho\right)}{p_{1}\left(1-q_{1}\right)g_{1}\left(\rho\right)+\left(1-p_{1}\right)\left(1-q_{0}\right)g_{1}\left(\rho\right)+p_{0}\left(1-q_{1}\right)g_{0}\left(\rho\right)+\left(1-p_{0}\right)\left(1-q_{0}\right)g_{0}\left(\rho\right)}
=\dfrac{\frac{\left(1-q_{0}\right)-p_{1}\left(q_{1}-q_{0}\right)}{\left(1-q_{0}\right)+p_{0}\left(q_{1}-q_{0}\right)}\frac{g_{1}\left(\rho\right)}{g_{0}\left(\rho\right)}}{\frac{\left(1-q_{0}\right)-p_{1}\left(q_{1}-q_{0}\right)}{\left(1-q_{0}\right)+p_{0}\left(q_{1}-q_{0}\right)}\frac{g_{1}\left(\rho\right)}{g_{0}\left(\rho\right)}+1}.
\end{array}
\]

Suppose instead that the signal is not garbled, but troll farm operates. Take a sender's strategy under which with probability $1-\alpha$ the original signal $s$ is observed, while with probability $\alpha$, signal $s$ is replaced by signal $\tilde{s}$ which equals $1$ with probability $\tilde{p}$ and $0$ with probability $1-\tilde{p}$. For a given realisation of $\rho$, let $\pi^{T}\left(1,\rho\right)$ and $\pi^{T}\left(0,\rho\right)$ be the voter's posterior beliefs when trolls are operating and the voter observes, respectively, signals $s1$ and $0$ as well as signal $\rho$. These beliefs are:
\[
\pi^{T}\left(1,\rho\right)=\frac{\left(1-\alpha\right)p_{1}g_{1}\left(\rho\right)+\alpha\tilde{p}g_{1}\left(\rho\right)}{\left(1-\alpha\right)p_{1}g_{1}\left(\rho\right)+\alpha\tilde{p}g_{1}\left(\rho\right)+\left(1-\alpha\right)p_{0}g_{0}\left(\rho\right)+\alpha\tilde{p}g_{0}\left(\rho\right)}=\frac{\frac{\left(1-\alpha\right)p_{1}+\alpha\tilde{p}}{\left(1-\alpha\right)p_{0}+\alpha\tilde{p}}\frac{g_{1}\left(\rho\right)}{g_{0}\left(\rho\right)}}{\frac{\left(1-\alpha\right)p_{1}+\alpha\tilde{p}}{\left(1-\alpha\right)p_{0}+\alpha\tilde{p}}\frac{g_{1}\left(\rho\right)}{g_{0}\left(\rho\right)}+1},
\]
and
\[\begin{array}{c}
\pi^{T}\left(0,\rho\right)=\frac{\left(1-\alpha\right)\left(1-p_{1}\right)g_{1}\left(\rho\right)+\alpha\left(1-\tilde{p}\right)g_{1}\left(\rho\right)}{\left(1-\alpha\right)\left(1-p_{1}\right)g_{1}\left(\rho\right)+\alpha\left(1-\tilde{p}\right)g_{1}\left(\rho\right)+\left(1-\alpha\right)\left(1-p_{0}\right)g_{0}\left(\rho\right)+\alpha\left(1-\tilde{p}\right)g_{0}\left(\rho\right)}=\frac{\frac{\left(1-\alpha\right)\left(1-p_{1}\right)+\alpha\left(1-\tilde{p}\right)}{\left(1-\alpha\right)\left(1-p_{0}\right)+\alpha\left(1-\tilde{p}\right)}\frac{g_{1}\left(\rho\right)}{g_{0}\left(\rho\right)}}{\frac{\left(1-\alpha\right)\left(1-p_{1}\right)+\alpha\left(1-\tilde{p}\right)}{\left(1-\alpha\right)\left(1-p_{0}\right)+\alpha\left(1-\tilde{p}\right)}\frac{g_{1}\left(\rho\right)}{g_{0}\left(\rho\right)}+1}.
\end{array}
\]

To prove the result, it is sufficient to show that there exists a pair $\left(\alpha,\tilde{p}\right)\in\left[0,1\right]^{2}$ under which the distribution of the voter's posterior beliefs is the same as under garbling. That is, it is sufficient to find $\left(\alpha,\tilde{p}\right)$ such that for each $\rho$, either (i) $\pi^{T}\left(1,\rho\right)=\pi^{G}\left(s^{\prime}=1,\rho\right)$ and $\pi^{T}\left(0,\rho\right)=\pi^{G}\left(s^{\prime}=0,\rho\right)$; or $\pi^{T}\left(1,\rho\right)=\pi^{G}\left(s^{\prime}=0,\rho\right)$ and $\pi^{T}\left(0,\rho\right)=\pi^{G}\left(s^{\prime}=1,\rho\right)$.
Suppose $q_{1}\geq q_{0}$ and let $
\alpha=1-\left(q_{1}-q_{0}\right)\in\left(0,1\right)$
as well as $\tilde{p}=\frac{q_{0}}{\alpha}=\frac{q_{0}}{1-\left(q_{1}-q_{0}\right)}\in\left(0,1\right)$.
Then
\[
 \dfrac{\left(1-\alpha\right)p_{1}+\alpha\tilde{p}}{\left(1-\alpha\right)p_{0}+\alpha\tilde{p}}=\dfrac{p_{1}q_{1}+\left(1-p_{1}\right)q_{0}}{p_{0}q_{1}+\left(1-p_{0}\right)q_{0}}
\iff  \pi^{T}\left(1,\rho\right)=\pi^{G}\left(s^{\prime}=1,\rho\right).
\]
Furthermore,
\[
  \dfrac{\left(1-\alpha\right)\left(1-p_{1}\right)+\alpha\left(1-\tilde{p}\right)}{\left(1-\alpha\right)\left(1-p_{0}\right)+\alpha\left(1-\tilde{p}\right)}=\dfrac{1-q_{0}-p_{1}\left(q_{1}-q_{0}\right)}{1-q_{0}-p_{0}\left(q_{1}-q_{0}\right)}
\iff \pi^{T}\left(0,\rho\right)=\pi^{G}\left(s^{\prime}=0,\rho\right).
\]
Hence, the distributions of posterior beliefs are the same.
Suppose next that $q_{1}<q_{0}$. Let 
$
\alpha=1-\left(q_{0}-q_{1}\right)\in\left(0,1\right)$
and $\tilde{p}=\frac{1-q_{0}}{1-\left(q_{0}-q_{1}\right)}=\frac{1-q_{0}}{\alpha}\in\left(0,1\right)$. Then
\[
 \dfrac{\left(1-\alpha\right)p_{1}+\alpha\tilde{p}}{\left(1-\alpha\right)p_{0}+\alpha\tilde{p}}=\dfrac{1-q_{0}-p_{1}\left(q_{1}-q_{0}\right)}{1-q_{0}+p_{0}\left(q_{1}-q_{0}\right)}
\iff  \pi^{T}\left(0,\rho\right)=\pi^{G}\left(s^{\prime}=0,\rho\right).
\]
Furthermore,
\[
 \dfrac{\left(1-\alpha\right)\left(1-p_{1}\right)+\alpha\left(1-\tilde{p}\right)}{\left(1-\alpha\right)\left(1-p_{0}\right)+\alpha\left(1-\tilde{p}\right)}=\dfrac{p_{1}q_{1}+\left(1-p_{1}\right)q_{0}}{p_{0}q_{1}+\left(1-p_{0}\right)q_{0}}
\iff \pi^{T}\left(0,\rho\right)=\pi^{G}\left(s^{\prime}=1,\rho\right).
\]
Hence, the distributions of posterior beliefs are again the same.\qed

\subsection{Proof of Proposition \ref{prop:garblings_2_signals}}


As before, let $W\left(F,G\right)$ be the sender's expected payoff under information structure $\left(F,G\right)$.
Let $M\left(F\right)$ be the set of all garblings of $F$. Then the sender's equilibrium expected payoff is $\max_{\hat{F}\in M\left(F\right)}W\left(\hat{F},G\right).$ Suppose $F$ becomes more informative, that is, is replaced by $F^{\prime}$ such that $F$ is a garbling of $F^{\prime}$. Then $M\left(F\right)\subset M\left(F^{\prime}\right)$. Therefore,
\[
\max_{\hat{F}\in M\left(F^{\prime}\right)}W\left(\hat{F},G\right)\geq\max_{\hat{F}\in M\left(F\right)}W\left(\hat{F},G\right),
\]
so the sender's equilibrium expected payoff increases.\qed

\subsection{Proof of Proposition \ref{prop: nonbayesian}}

Using the same steps as in Proposition \ref{lem: strategy}, with $s$ replaced by $\beta\left(s\right)$, we can show that at the equilibrium, the share of voters voting for the government in each state  equals
\[
V_{0}=H\left(\frac{1}{2}\right)+\int_{\frac{1}{2}}^{1}\frac{\frac{F_{0}\left(\beta^{-1}\left[s^{*}\left(x\right)\right]\right)}{F_{1}\left(\beta^{-1}\left[s^{*}\left(x\right)\right]\right)}-1}{\frac{x}{1-x}\frac{F_{0}\left(\beta^{-1}\left[s^{*}\left(x\right)\right]\right)}{F_{1}\left(\beta^{-1}\left[s^{*}\left(x\right)\right]\right)}-1}dH\left(x\right),
\]
and
\[
V_{1}=H\left(\frac{1}{2}\right)+\int_{\frac{1}{2}}^{1}\frac{\frac{F_{0}\left(\beta^{-1}\left[s^{*}\left(x\right)\right]\right)}{F_{1}\left(\beta^{-1}\left[s^{*}\left(x\right)\right]\right)}-1}{\frac{F_{0}\left(\beta^{-1}\left[s^{*}\left(x\right)\right]\right)}{F_{1}\left(\beta^{-1}\left[s^{*}\left(x\right)\right]\right)}-\frac{1-x}{x}}dH\left(x\right).
\]

Recall that $s^{*}\left(x\right)>0$ for all $x>\frac{1}{2}$. If $\hat{\beta}$ is more conservative than $\beta$, then $\hat{\beta}^{-1}\left[s^{*}\left(x\right)\right]>\beta^{-1}\left[s^{*}\left(x\right)\right]$ for all $x>\frac{1}{2}$. As $\frac{F_0\left(z\right)}{F_1\left(z\right)}$ is decreasing in $z$ for all $z>0$ by MLRP (see the proof of Theorem \ref{prop: polarisation}), this implies that $V_{0}$ and $V_{1}$ are both smaller under  $\hat{\beta}$ than under $\beta$. \qed

\subsection{Proof of Proposition \ref{pro:naive}}
\label{seca:naive}
It follows from the text that the sender can guarantee that every pro-government voter votes for the government.
Thus, consider anti-government voters now. 
For  $s \ge s^*(x)$, all naive voters support the government for all
$\alpha_x \in [0,1]$. Without changing slant, increasing $\alpha_x$ above
$\alpha_x^*$ as defined in \eqref{eq:alpha_x>1/2} induces non-naive  voters to oppose the government. If $\alpha_x > \alpha_x^*$, slant cannot be adjusted so that
all $s \ge s^*(x)$ induce non-naive  voters to support the government.
Hence, increasing $\alpha_x$ beyond $\alpha_x^*$ requires a change in slant and will lose the votes of some non-naive voters.

Note that it can  never be optimal to have $\pi_x(s) > x$, since then  $\alpha_x$ could be increased or slant
reallocated profitably. Thus, define $\mathcal{Y}$ and $\mathcal{N}$ such that  $\pi_x(s)=x$ for $s \in \mathcal{Y}$ and
$\pi_x(s)<x$ for
$\mathcal{N}:=[s^*(x),\infty)\setminus\mathcal{Y}$, and denote by 
$\left|\mathcal{N}\right|$ the Lebesgue measure of $\mathcal{N}$ and by $\lambda(\alpha_x)=\mathbb{P}(s\in \mathcal{N})$ the probability with which the troll farm chooses $s\in \mathcal{N}$. For a given $\lambda(\alpha_x)$, if $\left|\mathcal{N}\right|>0$, then for  both troll and genuine signals, the probability that voters with $s\ge s^*(x)$ oppose the government is 
positive.
Concentrating   $\lambda(\alpha_x)$ on a single realization
$s'\ge s^*(x)$ instead makes support after a genuine signal certain, implying
that at the optimum $\left|\mathcal{N}\right|=0$. Thus, 
using 
\eqref{eq:f-tilde},
\[
\int_{s^*(x)}^{\infty} \tilde f(s)\, ds
=
\frac{(1-\alpha_x)}{\alpha_x}\int_{s^*(x)}^{\infty} \kappa'(s)ds
= 1 - \lambda\left(\alpha_x\right)\iff \lambda\left(\alpha_x\right)=1-\frac{1-\alpha_x}{\alpha_x}\int_{s^*(x)}^{\infty} \kappa'(s)ds,
\]
where $\kappa(s)$ is defined in \eqref{eq:kappa}. Note that $\lambda\left(\alpha_x^*\right)=0$ and $\lambda\left(1\right)=1$.
The government's vote share in state
$\theta$ is then
\[
V_\theta=\phi\left[\alpha_x+(1-\alpha_x)(1-F_\theta[s^*(x)])\right]+(1-\phi)\left[\alpha_x (1-\lambda(\alpha_x))+(1-\alpha_x)(1-F_\theta[s^*(x)])\right],
\]
which  is linear in $\alpha_x$. The optimum is either
$\alpha_x=\alpha_x^*$ or $\alpha_x\rightarrow 1$, with the latter being optimal iff $\phi$ is
sufficiently large. 

Hence, the government's vote share among voters of type $x>\frac{1}{2}$ in state $\theta$ equals $\max\left\{\phi,V_{\theta}\left(x\right)\right\}$, where $V_{\theta}\left(x\right)$ is defined as in \eqref{eq:vote share x}. The government's overall vote share in state $\theta$ equals $H\left(\frac{1}{2}\right)+\int_{\frac{1}{2}}^{1}\max\left\{\phi,V_{\theta}\left(x\right)\right\}dH\left(x\right)$. Using the same logic as in the proofs of  Theorems \ref{prop: informativeness_Blackwell} and  \ref{prop: polarisation}, we can show the result.\qed




%
%
%
%
%

\bibliography{trollfarms}{}

@article{augenblick2025overinference,
  title={Overinference from weak signals and underinference from strong signals},
  author={Augenblick, Ned and Lazarus, Eben and Thaler, Michael},
  journal={The Quarterly Journal of Economics},
  volume={140},
  number={1},
  pages={335--401},
  year={2025},
  publisher={Oxford University Press}
}

@article{BergemannMorris:2016,
Author = {Bergemann, Dirk and Morris, Stephen},
Title = {Information Design, Bayesian Persuasion, and Bayes Correlated Equilibrium},
Journal = {American Economic Review},
Volume = {106},
Number = {5},
Year = {2016},
Month = {May},
Pages = {586-91},
DOI = {10.1257/aer.p20161046},
URL = {https://www.aeaweb.org/articles?id=10.1257/aer.p20161046}}

@article{Matyskova:2018,
title = {Bayesian persuasion with costly information acquisition},
journal = {Journal of Economic Theory},
volume = {211},
pages = {105678},
year = {2023},
author = {Ludmila Matyskov\'{a} and Alfonso Montes} 
}

@inproceedings{blackwell1951comparison,
  title={Comparison of experiments},
  author={David Blackwell},
  booktitle={Proceedings of the Second Berkeley Symposium on Mathematical Statistics and Probability},
  volume={1},
  pages={93--102},
  year={1951}
}

@article{zhuravskaya2020political,
  title={Political effects of the internet and social media},
  author={Zhuravskaya, Ekaterina and Petrova, Maria and Enikolopov, Ruben},
  journal={Annual Review of Economics},
  volume={12},
  pages={415--438},
  year={2020},
  publisher={Annual Reviews}
}

@Article{dentersocial,
    author = {Denter, Philipp and Dumav, Martin and Ginzburg, Boris},
    title = "{Social Connectivity, Media Bias, and Correlation Neglect}",
    journal = {The Economic Journal},
    year = {2021},
    volume={131},
    pages={2033--2057},
    doi = {10.1093/ej/ueaa128}
}

@article{ortoleva2024alternatives,
  title={Alternatives to Bayesian updating},
  author={Ortoleva, Pietro},
  journal={Annual Review of Economics},
  volume={16},
  number={1},
  pages={545--570},
  year={2024},
  publisher={Annual Reviews}
}

@article{pons2018expressive,
  title={Expressive voting and its cost: Evidence from runoffs with two or three candidates},
  author={Pons, Vincent and Tricaud, Cl{\'e}mence},
  journal={Econometrica},
  volume={86},
  number={5},
  pages={1621--1649},
  year={2018},
  publisher={Wiley Online Library}
}

@article{heese2025persuasion,
  title={Persuasion and information aggregation in elections},
  author={Heese, Carl and Lauermann, Stephan},
  journal={Journal of Political Economy},
  volume={133},
  number={10},
  pages={3305--3348},
  year={2025},
  publisher={The University of Chicago Press Chicago, IL}
}

@Article{guardian_sock_puppet_software,
  author = {{The Guardian}},
  title  = {Revealed: US spy operation that manipulates social media},
  year   = {2011},
  note   = {Accessed on 19 April 2021},
}

@Article{freedomhouse2017automated,
  author = {{Freedom House}},
  title  = {Manipulating Social Media to Undermine Democracy},
  year   = {2017},
  note   = {Accessed on 26 April 2021},
}

@Article{freedomhouse2019crisis,
  author = {{Freedom House}},
  title  = {The Crisis of Social Media},
  year   = {2019},
  note   = {Accessed on 26 April 2021},
}

@article{tversky1992advances,
  title={Advances in prospect theory: Cumulative representation of uncertainty},
  author={Tversky, Amos and Kahneman, Daniel},
  journal={Journal of Risk and Uncertainty},
  volume={5},
  number={4},
  pages={297--323},
  year={1992},
  publisher={Springer}
}

@article{ederer2022bayesian,
author={Florian Ederer and Weicheng Min},
title={Bayesian persuasion with lie detection},
year={2022},
journal={NBER Working Paper 30065}
}

@article{Leber:2012,
author={Jessica Leber},
year={2012},
title={Campaigns to Track Voters with ``Political Cookies''},
journal={MIT Technology Review}
}

@article{KamenicaGentzkow:2011,
Author = {Kamenica, Emir and Gentzkow, Matthew},
Title = {Bayesian Persuasion},
Journal = {American Economic Review},
Volume = {101},
Number = {6},
Year = {2011},
Pages = {2590-2615},
DOI = {10.1257/aer.101.6.2590},
URL = {http://www.aeaweb.org/articles.php?doi=10.1257/aer.101.6.2590}}

@article{martherus2021party,
  title={Party animals? Extreme partisan polarization and dehumanization},
  author={Martherus, James L and Martinez, Andres G and Piff, Paul K and Theodoridis, Alexander G},
  journal={Political Behavior},
  volume={43},
  number={2},
  pages={517--540},
  year={2021},
  publisher={Springer}
}

@article{chan2019pivotal,
  title={Pivotal persuasion},
  author={Chan, Jimmy and Gupta, Seher and Li, Fei and Wang, Yun},
  journal={Journal of Economic Theory},
  volume={180},
  pages={178--202},
  year={2019},
  publisher={Elsevier}
}

@article{bardhi2018modes,
  title={Modes of persuasion toward unanimous consent},
  author={Bardhi, Arjada and Guo, Yingni},
  journal={Theoretical Economics},
  volume={13},
  number={3},
  pages={1111--1149},
  year={2018},
  publisher={Wiley Online Library}
}

@article{mccoy2018polarization,
  title={Polarization and the global crisis of democracy: Common patterns, dynamics, and pernicious consequences for democratic polities},
  author={McCoy, Jennifer and Rahman, Tahmina and Somer, Murat},
  journal={American Behavioral Scientist},
  volume={62},
  number={1},
  pages={16--42},
  year={2018},
  publisher={Sage Publications Sage CA: Los Angeles, CA}
}

@article{cai2006overcommunication,
  title={Overcommunication in strategic information transmission games},
  author={Cai, Hongbin and Wang, Joseph Tao-Yi},
  journal={Games and Economic Behavior},
  volume={56},
  number={1},
  pages={7--36},
  year={2006},
  publisher={Elsevier}
}

@article{jin2021no,
  title={Is no news (perceived as) bad news? An experimental investigation of information disclosure},
  author={Jin, Ginger Zhe and Luca, Michael and Martin, Daniel},
  journal={American Economic Journal: Microeconomics},
  volume={13},
  number={2},
  pages={141--73},
  year={2021}
}

@article{patty2007letting,
  title={Letting the good times roll: A theory of voter inference and experimental evidence},
  author={Patty, John W and Weber, Roberto A},
  journal={Public Choice},
  volume={130},
  number={3},
  pages={293--310},
  year={2007},
  publisher={Springer}
}

@article{ginzburg2019collective,
  title={When collective ignorance is bliss: Theory and experiment on voting for learning},
  author={Ginzburg, Boris and Guerra, Jos{\'e}-Alberto},
  journal={Journal of Public Economics},
  volume={169},
  pages={52--64},
  year={2019},
  publisher={Elsevier}
}

@article{goldstein2023generative,
  title={Generative Language Models and Automated Influence Operations: Emerging Threats and Potential Mitigations},
  author={Goldstein, Josh A and Sastry, Girish and Musser, Micah and DiResta, Renee and Gentzel, Matthew and Sedova, Katerina},
  journal={arXiv preprint arXiv:2301.04246},
  year={2023}
}

@article{gorodnichenko2021social,
  title={Social media, sentiment and public opinions: Evidence from \#Brexit and \#USElection},
  author={Gorodnichenko, Yuriy and Pham, Tho and Talavera, Oleksandr},
  journal={European Economic Review},
  volume={136},
  pages={103772},
  year={2021},
  publisher={Elsevier}
}

@article{king2017chinese,
  title={How the Chinese government fabricates social media posts for strategic distraction, not engaged argument},
  author={King, Gary and Pan, Jennifer and Roberts, Margaret E},
  journal={American Political Science Review},
  volume={111},
  number={3},
  pages={484--501},
  year={2017},
  publisher={Cambridge University Press}
}

@article{ginzburg2019optimal,
  title={Optimal information censorship},
  author={Ginzburg, Boris},
  journal={Journal of Economic Behavior \& Organization},
  volume={163},
  pages={377--385},
  year={2019},
  publisher={Elsevier}
}

@article{gradwohl2022social,
  title={How social media can undermine democracy},
  author={Gradwohl, Ronen and Heller, Yuval and Hillman, Arye},
  journal={European Journal of Political Economy},
  volume={86},
  pages={102634},
  year={2025},
  publisher={Elsevier}
}

@article{kolotilin2022censorship,
  title={Censorship as optimal persuasion},
  author={Kolotilin, Anton and Mylovanov, Timofiy and Zapechelnyuk, Andriy},
  journal={Theoretical Economics},
  volume={17},
  number={2},
  pages={561--585},
  year={2022},
  publisher={Wiley Online Library}
}

@article{edmond2021creating,
  title={Creating confusion},
  author={Edmond, Chris and Lu, Yang K},
  journal={Journal of Economic Theory},
  volume={191},
  pages={105145},
  year={2021},
  publisher={Elsevier}
}

@article{edmond2013information,
  title={Information manipulation, coordination, and regime change},
  author={Edmond, Chris},
  journal={Review of Economic Studies},
  volume={80},
  number={4},
  pages={1422--1458},
  year={2013},
  publisher={Oxford University Press}
}

@article{onuchic2023conveying,
  title={Conveying value via categories},
  author={Onuchic, Paula and Ray, Debraj},
  journal={Theoretical Economics},
  volume={18},
  number={4},
  pages={1407--1439},
  year={2023},
  publisher={Wiley Online Library}
}

@article{di2021strategic,
  title={Strategic sample selection},
  author={Di Tillio, Alfredo and Ottaviani, Marco and S{\o}rensen, Peter Norman},
  journal={Econometrica},
  volume={89},
  number={2},
  pages={911--953},
  year={2021},
  publisher={Wiley Online Library}
}

@article{le2019persuasion,
  title={Persuasion with limited communication capacity},
  author={{Le Treust}, Ma{\"e}l and Tomala, Tristan},
  journal={Journal of Economic Theory},
  volume={184},
  pages={104940},
  year={2019},
  publisher={Elsevier}
}

@article{escude2023slow,
  title={Slow persuasion},
  author={Escud{\'e}, Matteo and Sinander, Ludvig},
  journal={Theoretical Economics},
  volume={18},
  number={1},
  pages={129--162},
  year={2023},
  publisher={Wiley Online Library}
}

@article{rossetti2023bots,
  title={Bots, disinformation, and the first impeachment of US President Donald Trump},
  author={Rossetti, Michael and Zaman, Tauhid},
  journal={PLoS ONE},
  volume={18},
  number={5},
  pages={e0283971},
  year={2023},
  publisher={Public Library of Science San Francisco, CA USA}
}

@article{acemoglu2024model,
  title={A model of online misinformation},
  author={Acemoglu, Daron and Ozdaglar, Asuman and Siderius, James},
  journal={Review of Economic Studies},
  volume={91},
  number={6},
  pages={3117--3150},
  year={2024},
  publisher={Oxford University Press UK}
}

@article{Linvill2020troll,
  title={Troll factories: Manufacturing specialized disinformation on Twitter},
  author={Linvill, Darren L and Warren, Patrick L},
  journal={Political Communication},
  volume={37},
  number={4},
  pages={447--467},
  year={2020},
  publisher={Taylor \& Francis}
}

@techreport{SisakDenter:2025,
author={Dana Sisak and Philipp Denter},
title={Truth, Lies, and Social Ties: When Image Concerns Fuel Fake News},
institution={arXiv:2410.19557},
year={2025}
}

@article{shen2023examining,
  title={Examining the differences between human and bot social media accounts: A case study of the Russia-Ukraine War},
  author={Shen, Fei and Zhang, Erkun and Zhang, Hongzhong and Ren, Wujiong and Jia, Quanxin and He, Yuan},
  journal={First Monday},
  year={2023}
}

@article{ba2023over,
  title={Over-and Underreaction to Information},
  author={Ba, Cuimin and Bohren, J Aislinn and Imas, Alex},
  year={2023},
journal={SSRN Working Paper 4274617}
}

@article{ferrara2016rise,
  title={The rise of social bots},
  author={Ferrara, Emilio and Varol, Onur and Davis, Clayton and Menczer, Filippo and Flammini, Alessandro},
  journal={Communications of the ACM},
  volume={59},
  number={7},
  pages={96--104},
  year={2016},
  publisher={ACM New York, NY, USA}
}

@article{SchipperWoo:2019,
year = {2019},
volume = {14},
journal = {Quarterly Journal of Political Science},
title = {Political Awareness, Microtargeting of Voters, and Negative Electoral Campaigning},
doi = {10.1561/100.00016066},
issn = {1554-0626},
number = {1},
pages = {41-88},
author = {Burkhard C. Schipper and Hee Yeul Woo}
}

@article{vanGils2025microtargeting,
  title={Microtargeting, voters' unawareness, and democracy},
  author={van Gils, Freek and M{\"u}ller, Wieland and Pr{\"u}fer, Jens},
  journal={The Journal of Law, Economics, and Organization},
  volume={41},
  number={2},
  pages={634--653},
  year={2025},
  publisher={Oxford University Press}
}

@article{ZaroualiEtAl:2022,
author = {Brahim Zarouali and Tom Dobber and Guy De Pauw and Claes de Vreese},
title ={Using a Personality-Profiling Algorithm to Investigate Political Microtargeting: Assessing the Persuasion Effects of Personality-Tailored Ads on Social Media},
journal = {Communication Research},
volume = {49},
number = {8},
pages = {1066-1091},
year = {2022},
doi = {10.1177/0093650220961965}
}

@article{YouyouEtAl:2016,
author = {Wu Youyou  and Michal Kosinski  and David Stillwell },
title = {Computer-based personality judgments are more accurate than those made by humans},
journal = {Proceedings of the National Academy of Sciences},
volume = {112},
number = {4},
pages = {1036-1040},
year = {2015},
doi = {10.1073/pnas.1418680112},
URL = {https://www.pnas.org/doi/abs/10.1073/pnas.1418680112},
eprint = {https://www.pnas.org/doi/pdf/10.1073/pnas.1418680112}
}

@article{SimchonEtAl:2024,
    author = {Almog Simchon  and  Matthew  Edwards and Stephan Lewandowsky},
    title = {The persuasive effects of political microtargeting in the age of generative artificial intelligence},
    journal = {PNAS Nexus},
    volume = {3},
    number = {2},
    pages = {1-5},
    year = {2024}
}

@article{TappinEtAl:2023,
author = {Ben M. Tappin  and Chloe Wittenberg  and Luke B. Hewitt  and Adam J. Berinsky  and David G. Rand },
title = {Quantifying the potential persuasive returns to political microtargeting},
journal = {Proceedings of the National Academy of Sciences},
volume = {120},
number = {25},
pages = {e2216261120},
year = {2023},
doi = {10.1073/pnas.2216261120},
URL = {https://www.pnas.org/doi/abs/10.1073/pnas.2216261120},
eprint = {https://www.pnas.org/doi/pdf/10.1073/pnas.2216261120}}

@article{haustein2016tweets,
 title={Tweets as impact indicators: Examining the implications of automated ``bot'' accounts on Twitter},
  author={Haustein, Stefanie and Bowman, Timothy D and Holmberg, Kim and Tsou, Andrew and Sugimoto, Cassidy R and Larivi{\`e}re, Vincent},
  journal={Journal of the Association for Information Science and Technology},
  volume={67},
  number={1},
  pages={232--238},
  year={2016},
  publisher={Wiley Online Library}
}

@Article{NYT:Cambridge,
  author = {{New York Times}},
  title  = {Cambridge Analytica and Facebook: The Scandal and the Fallout So Far},
  year   = {2018},
  note   = {Accessed on 11 April 2024}
}

@article{gehlbach2014government,
  title={Government control of the media},
  author={Gehlbach, Scott and Sonin, Konstantin},
  journal={Journal of Public Economics},
  volume={118},
  pages={163--171},
  year={2014},
  publisher={Elsevier}
}

@article{gitmez2023dictator,
  title={The Dictator's Dilemma: A Theory of Propaganda and Repression},
  author={Gitmez, A Arda and Sonin, Konstantin},
  year={2023},
journal={University of Chicago, Becker Friedman Institute for Economics Working Paper No. 2023-67}
}

@article{gitmez2023informational,
  title={Informational Autocrats, Diverse Societies},
  author={Gitmez, A Arda and Molavi, Pooya},
  year={2023},
journal={arXiv:2203.12698}
}

@article{boleslavsky2021media,
  title={Media Freedom in the Shadow of a Coup},
  author={Boleslavsky, Raphael and Shadmehr, Mehdi and Sonin, Konstantin},
  journal={Journal of the European Economic Association},
  volume={19},
  number={3},
  pages={1782--1815},
  year={2021},
  publisher={Oxford University Press}
}

@article{shadmehr2015state,
  title={State censorship},
  author={Shadmehr, Mehdi and Bernhardt, Dan},
  journal={American Economic Journal: Microeconomics},
  volume={7},
  number={2},
  pages={280--307},
  year={2015},
  publisher={American Economic Association 2014 Broadway, Suite 305, Nashville, TN 37203-2425}
}

@article{ginzburg2024comparing,
title = {A flexible measure of voter polarization},
journal = {Economics Letters},
volume = {261},
pages = {112834},
year = {2026},
issn = {0165-1765},
doi = {https://doi.org/10.1016/j.econlet.2026.112834},
url = {https://www.sciencedirect.com/science/article/pii/S0165176526000285},
author = {Boris Ginzburg},
}

@article{ginzburg2022counting,
  title={Counting on my vote not counting: Expressive voting in committees},
  author={Ginzburg, Boris and Guerra, Jos{\'e}-Alberto and Lekfuangfu, Warn N},
  journal={Journal of Public Economics},
  volume={205},
  pages={104555},
  year={2022},
  publisher={Elsevier}
}

@article{feddersen2009moral,
  title={Moral bias in large elections: Theory and experimental evidence},
  author={Feddersen, Timothy and Gailmard, Sean and Sandroni, Alvaro},
  journal={American Political Science Review},
  volume={103},
  number={2},
  pages={175--192},
  year={2009},
  publisher={Cambridge University Press}
}

@article{caselli2014incumbency,
  title={The incumbency effects of signalling},
  author={Caselli, Francesco and Cunningham, Tom and Morelli, Massimo and de Barreda, In{\'e}s Moreno},
  journal={Economica},
  volume={81},
  number={323},
  pages={397--418},
  year={2014},
  publisher={Wiley Online Library}
}

@article{akoz2016information,
  title={Information Manipulation in Election Campaigns},
  author={Ak{\"o}z, Kemal K{\i}van{\c{c}}  and Arbatli, Cemal Eren},
  journal={Economics \& Politics},
  volume={28},
  number={2},
  pages={181--215},
  year={2016},
  publisher={Wiley Online Library}
}

@article{simchon2022troll,
  title={Troll and divide: the language of online polarization},
  author={Simchon, Almog and Brady, William J and Van Bavel, Jay J},
  journal={PNAS nexus},
  volume={1},
  number={1},
  pages={pgac019},
  year={2022},
  publisher={Oxford University Press}
}

@article{linvill2022talking,
  title={Talking to trolls---how users respond to a coordinated information operation and why they're so supportive},
  author={Linvill, Darren L and Warren, Patrick L and Moore, Amanda E},
  journal={Journal of Computer-Mediated Communication},
  volume={27},
  number={1},
  pages={zmab022},
  year={2022},
  publisher={Oxford University Press}
}

@article{foerster2023theory,
  title={A theory of media bias and disinformation},
  author={Foerster, Manuel},
  year={2023}
}

@article{aymanns2025fake,
  title={Fake news in social networks},
  author={Aymanns, Christoph and Foerster, Jakob and Georg, Co-Pierre and Weber, Matthias},
  year={2025}
}
\bibliographystyle{aer}

\end{document}